\documentclass[%
 reprint,
superscriptaddress,
 amsmath,amssymb,
 aps,
pra,
]{revtex4-2}

\usepackage[ruled,vlined]{algorithm2e}
\usepackage{url}
\usepackage{amsfonts,amsmath, amsthm, amscd,amssymb,amsbsy,amstext,amsopn,amsxtra}
\usepackage{hyperref}
\usepackage{float}  
\usepackage{subcaption}
\usepackage{tikz}
\newtheorem{thm}{Theorem}

\usepackage{graphicx}
\usepackage{dcolumn}
\usepackage{bm}

\begin{document}

\preprint{APS/123-QED}

\title{Systematic Approach to Hyperbolic Quantum Error Correction Codes}

\author{Ahmed Adel Mahmoud}
\email{ahmed.mahmoud@usask.ca}
\affiliation{Centre for Quantum Topology and Its Applications (quanTA), University of Saskatchewan, Saskatoon, SK, Canada}
\affiliation{Department of Mathematics and Statistics, University of Saskatchewan, Saskatoon, SK, Canada}

\author{Kamal Mohamed Ali}
\affiliation{Centre for Quantum Topology and Its Applications (quanTA), University of Saskatchewan, Saskatoon, SK, Canada}

\author{Steven Rayan}
\email{rayan@math.usask.ca}
\affiliation{Centre for Quantum Topology and Its Applications (quanTA), University of Saskatchewan, Saskatoon, SK, Canada}
\affiliation{Department of Mathematics and Statistics, University of Saskatchewan, Saskatoon, SK, Canada}

\begin{abstract}
Quantum error correction codes defined on hyperbolic lattices leverage the unique geometric properties of the hyperbolic space to enhance the performance of quantum error correction. By embedding qubits in hyperbolic lattices, these codes achieve higher encoding rates and lower qubit overhead compared to those defined on conventional Euclidean lattices. Building on recent advances in hyperbolic crystallography, we introduce a unified framework for the systematic construction and scalable benchmarking of CSS quantum error correction codes on hyperbolic lattices. A central component of this framework is the Hyperbolic Cycle Basis algorithm, which employs graph-theoretic methods to efficiently identify all plaquette cycles (parity-check supports) and nontrivial cycles (logical operators). This enables scalable and automated benchmarking of a broad class of CSS codes defined on hyperbolic geometries. We apply this framework to construct and simulate two representative hyperbolic quantum error correction codes (HQECCs), evaluating key performance metrics such as encoding rate, error threshold, and code distance for different sublattices. While HQECCs serve as concrete examples, the framework can be adapted to a wide range of CSS codes, including those with more intricate stabilizer structures such as Floquet codes. This work establishes a foundation for systematic exploration and benchmarking of CSS codes on hyperbolic lattices, paving the way toward practical, high-performance quantum error correction.
\end{abstract}


\maketitle

\section{Introduction}
Quantum error correction is a prerequisite for the scalable and practical implementation of quantum computing, ensuring computational reliability in the presence of noise and decoherence \cite{gottesman2009introduction}. Among the various families of quantum error correction codes, quantum low-density parity-check (LDPC) codes have gained significant attention due to their favorable properties. In particular, their parity-check operators act on a constant number of qubits, and each qubit participates in a constant number of parity-check measurements, making them highly efficient for fault-tolerant quantum computation \cite{terhal2015quantum, breuckmann2021quantum}. Over the past few years, several LDPC codes have been developed, with the toric code, introduced by Kitaev \cite{kitaev2003fault}, being one of the most well-known examples. The toric code offers high error tolerance but protects only two logical qubits regardless of lattice size \cite{dennis2002topological}. This limitation was overcome by extending the framework to hyperbolic lattices, leading to the discovery of HQECCs \cite{breuckmann2016constructions, breuckmann2017hyperbolic, albuquerque2009topological, kim2007quantum} Like the toric code, HQECCs offer high error tolerance but with the crucial advantage that the number of logical qubits grows linearly with physical qubits. This advantage of hyperbolic lattices in topological quantum error correction is further confirmed by state-of-the-art Floquet codes, which achieve
better performance and lower qubit overhead on hyperbolic lattices than on their original Euclidean honeycomb design \cite{hastings2021dynamically, gidney2021honeycomb, higgott2024constructions}. These advances established a clear paradigm: hyperbolic lattices are a fundamentally better foundation
for topological quantum error correction, yet they expose a critical gap between theoretical advances in topological quantum error correction and practical quantum hardware. The design space of hyperbolic lattices is infinitely vast compared to Euclidean alternatives \cite{ratcliffe2006foundations}. Yet, out of this immense family of lattices, it is not known which hyperbolic lattice is optimal for a given topological quantum error correction code and noise profile. Answering this question is essential not only for benchmarking current quantum hardware but also for guiding the design of next-generation quantum processors optimized for topological quantum error correction. To answer this question, a unified framework for constructing and benchmarking topological codes on hyperbolic lattices is needed. In this work, we present a systematic framework for constructing and benchmarking CSS codes on hyperbolic lattices, building on recent advances in hyperbolic crystallography \cite{boettcher2022crystallography}. This framework is applicable to any hyperbolic $\{p,q\}$ tessellation of the Poincar\'e disk with an underlying $\{p_B,q_B\}$ Bravais lattice. A central component of this framework is the Hyperbolic Cycle Basis algorithm, which, to the best of our knowledge, is the first systematic and scalable method for identifying all plaquette cycles (representing parity-check supports) in a hyperbolic tessellation, as well as all non-trivial cycles (representing logical operators) that generate the first homology group $H_1(M)$ of the underlying Riemann surface $M$.

To demonstrate the versatility of this framework, we apply it to simulate two HQECCs based on two hyperbolic tessellations of the Poincar\'e disk, namely $\{8,3\}$ and $\{10,3\}$. These tessellations correspond to two different Bravais lattice structures, each associated with a unique Fuchsian group. The unit cells of these lattices are embedded in genus-2 Riemann surfaces, obtained by compactifying the corresponding Bravais lattices $\{8,8\}$ and $\{10,5\}$, respectively.

Finally, even though this work focuses mainly on quantum error correction codes, the Hyperbolic Cycle Basis algorithm also has useful applications in condensed matter physics. In the spectral analysis of hyperbolic materials derived from kagome-like lattices \cite{kollar2019hyperbolic}, it has been shown that the ground state exihibts a highly degenerate flat band whose eigenstates are highly localized and therefore dubbed compact localized states \cite{Bzdusek2022, Kollar2020}. Each of these states is supported exclusively on an independent cycle—either trivial or non-trivial—of the underlying lattice.  As a result, the full set of compact localized states naturally forms a cycle basis of the underlying hyperbolic graph, a structure that can be efficiently identified using the HCB algorithm.

\section{Mathematical Preliminaries}
\subsection{Hyperbolic Geometry}\label{HG}
Two models of hyperbolic geometry are equivalent. The first model is the upper-half plane $\mathbb{H} = (z \in \mathbb{C} :Im(z) >0 )$ with boundary $\partial \mathbb{H} = \mathbb{R} \cup  \{\infty$\}. The second model is the Poincar\'e disk model $\mathbb{D} = \{z' \in \mathbb{C} : |z'|<1\}$ with boundary $\partial \mathbb{D} = \{z \in \mathbb{C} :  |z| =1  \}$. The hyperbolic space $\mathbb{H}$ is equipped with the hyperbolic metric
\begin{equation}
    ds^2 = \frac{dx^2 + dy^2}{y^2} .
\end{equation}
An isometry between the two spaces $h: \mathbb{H} \rightarrow \mathbb{D}$ is given by
\begin{equation}
    h(z) = \frac{zi+1}{z+i},
\end{equation}
where $z=x+iy \in \mathbb{C}$. The metric induced on $\mathbb{D}$ by $h$ is given by
\begin{equation}
    ds^2 = \frac{4|dz|^2}{(1 - |z|^2)^2}.
\end{equation}
Since the two models are equivalent, one can work in either model. However, since the Poincar\'e disk model is a bounded subset of the Euclidean plane, it is more convenient for visualization. Therefore, we will utilize the Poincar\'e disk model throughout this paper. 

The hyperbolic distance between any two points $z_1,z_2 \in \mathbb{D}$ is given by
\begin{equation}
    d(z_1,z_2) = \mbox{arcosh} \left(1 + \frac{2 |z_1 - z_2|^2}{(1-|z_1|^2)(1-|z_2|^2)}\right).
\end{equation}
Geodesics in $\mathbb{D}$ are circles that, when extended, are orthogonal to the boundary of the Poincar\'e disk and its diameter. The hyperbolic angle between two geodesics that intersect at a point is the usual Euclidean angle between the tangent vectors to these two geodesics. A hyperbolic polygon with $p$ edges, called a $p$-gon, is a convex closed set consisting of $p$ hyperbolic geodesic edges. The point at which two segments intersect is a vertex. A polygon is called regular if all its internal angles are equal. A regular tessellation of $\mathbb{D}$ is achieved by covering the Poincar\'e disk by regular $p$-gons that either do not overlap or overlap only at their boundaries. For convenience, we refer to regular tessellations as \emph{patterns}. Formally, a pattern is a finite hyperbolic graph embedded into a closed Riemann surface that determines the shape and size of the unit cell of a hyperbolic lattice. The Schläfli symbol of a pattern is $\{p,q\}$ if each face is a $p$-gon, and each vertex is surrounded by $q$ faces. To construct a $\{p,q\}$ pattern, one starts by constructing one polygon representing the unit cell of the pattern. Let $r$ be the radius of the polygon given by

\begin{equation}
r=\sqrt{\frac{\cos \left(\frac{\pi}{p}+\frac{\pi}{q}\right)}{\cos \left(\frac{\pi}{p}-\frac{\pi}{q}\right)}}.
\end{equation}
Then, the positions of the unit cell vertices in $\mathbb{D}$ are given by
\begin{equation}
    z_k = r e^{2 \pi ik/ p+\delta},
\end{equation}
where $k = 1,...,p$ and $\delta$ is an arbitrary phase.

Upon imposing periodic boundary conditions (PBCs), a hyperbolic $\{p,q\}$ pattern can be embedded in a closed Riemann surface of genus $g \geq2$ \cite{beardon2012geometry, katok1992fuchsian}. In this setting, the following combinatorial relations hold. Since each face contains $p$ edges and every edge is shared by two faces, counting face–edge incidences gives $pF=2E$. Likewise, since $q$ edges meet at every vertex and each edge connects two vertices, counting edge-vertex incidences gives $2E=qV$. By combining these relations, we obtain
\begin{equation}
    pF = 2 E = qV,
\end{equation}
where $F,E$ and $V$ are the number of faces, edges and vertices of the pattern respectively.
An important topological invariant that we shall use later is the Euler characteristic. Given a closed Riemann surface $M$ tessellated by $F$ faces, $E$ edges and $V$ vertices, the Euler characteristic is given by
\begin{equation}
    \chi(M) = F - E + V.
\end{equation}
If $\chi$ is even, then the tessellation can be embedded in an orientable surface $M$ of genus $g$; in this case \cite{stillwell1995geometry}
\begin{equation}
    \chi(M) = 2-2g.
\end{equation}
It follows that the number of faces $F$ and the genus $g$ of the closed Riemann surface are not independent. More concretely, consider the Gauss-Bonnet theorem that relates the curvature of a surface to its topology. It states that for a closed, orientable surface $M$, the genus $g$ of $M$ is proportional to its Gaussian curvature as follows \cite{do2016differential}
$$\int_M K \ dA = 4 \pi (1-g).$$
In the case of hyperbolic geonmetry $K=-1$; therefore, the area of the underlying Riemann surface $M$ is given by
\begin{equation}
    A(M)= 4\pi (g-1).
\end{equation}
\begin{figure*}[ht]
    \centering
    \begin{minipage}{0.35\linewidth}
        \centering
        \includegraphics[width=\linewidth]{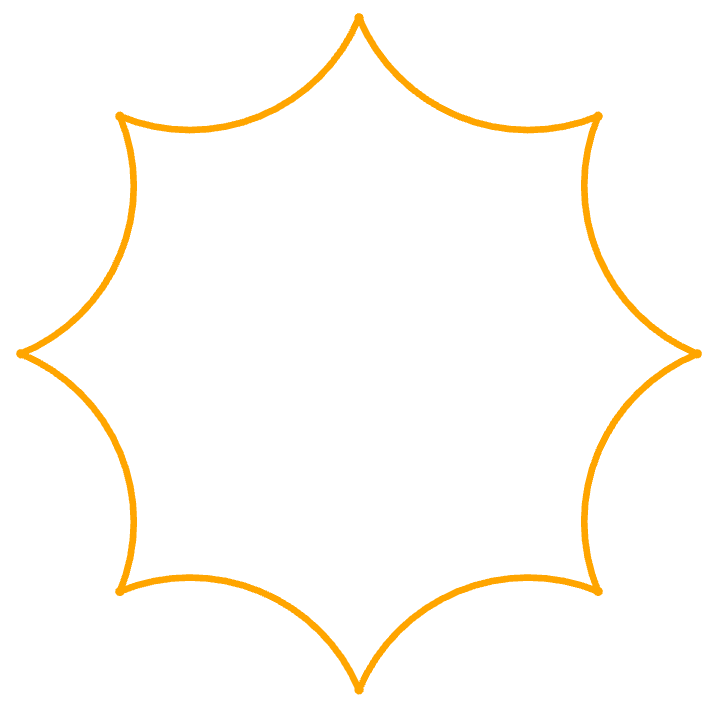}
    \end{minipage}
    \begin{minipage}{0.05\linewidth}
        \centering
        \begin{tikzpicture}
            \draw[thick,->, color=orange] (0,0) -- (1.5,0);
        \end{tikzpicture}
    \end{minipage}
    \begin{minipage}{0.55\linewidth}
        \centering
        \includegraphics[width=\linewidth]{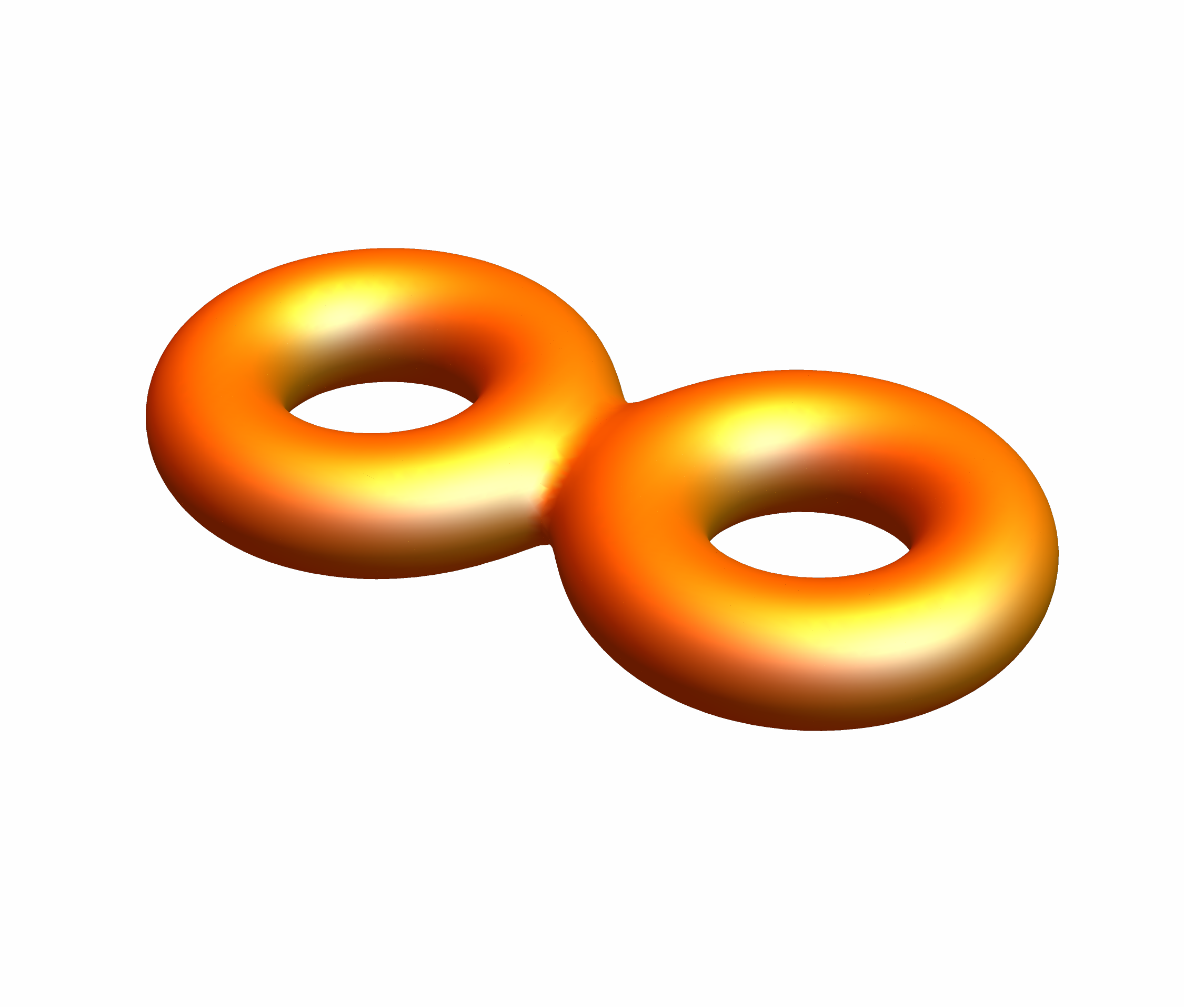}
    \end{minipage}

    \caption{The left half of the figure shows the unit cell of the $\{8,8\}$ Bravais lattice, which, when compactified, can be embedded in the genus-$2$ Riemann surface shown on the right.}
    \label{fig:polygon-riemann}
\end{figure*}

In hyperbolic geometry, the sum of the angles of a hyperbolic triangle is less than $\pi$. A special case of the Gauss-Bonnet theorem states that the area of a hyperbolic triangle $\Delta$ with internal angles $\alpha$, $\beta$, and $\zeta$ is given by
    \begin{equation}
        A(\Delta) = \pi - (\alpha + \beta + \zeta). 
    \end{equation}
Generally, the area of a regular hyperbolic $p$-gon $\Lambda$ is given by its angular defect
\begin{equation}
    A(\Lambda)= (p-2)\pi - \sum_{j=1}^p \alpha_j ,
\end{equation}
where $\alpha_j$ are the internal angles of $\Lambda$. In a regular $\{p,q\}$ pattern, the internal angles of each $p$-gon are given by $2\pi/q$. Therefore, the area of a $p$-gon is given by
\begin{equation}
    A(\Lambda) = \left(p-2-\frac{2p}{q}\right)\pi.
\end{equation}
Now, assume that a closed Riemann surface $M$ is tessellated by a regular $\{p, q\}$ pattern; then, $A(M)=F A(\Lambda)$. That is,
\begin{equation}
    4\pi (g-1)= F \left(p-2-\frac{2p}{q}\right)\pi,
\end{equation}
where $F$ is the number of faces of the underlying tessellation of $M$. Solving for $F$, we get
\begin{equation}
    F = \frac{4 q (g-1).}{pq - 2q - 2p}.
\end{equation}
This is in contrast to the Euclidean case in which the number of faces of the tessellation is independent of the genus of the embedding surface.

Let $(F,E,V)$ be a solution to (7), then $(nF,nE,nV)$ is also a solution to (7) for some $n \in \mathbb{Z}^+$. In this case, however, $\chi \rightarrow n\chi$. Therefore, increasing the number of faces is equivalent to increasing the genus of the underlying Riemann surface. Thus, for every pattern $\{p,q\}$, there is a minimal solution $(F_m,E_m,V_m)$ to (7), with a minimal number of faces $F_m$. 
An interesting class of patterns is one that satisfies $F_m = 1$. In this case, PBCs can be consistently defined for the associated polygon. That is, the edge pairing transformations are sufficient to embed the polygon into a closed Riemann surface of genus $g \geq 2$. In this case, $$(F,E,V)=(1,p/2,p/q);$$ 
therefore, $p$ is even and $p \geq q$. Four prominent infinite families of lattices satisfying this condition were presented in \cite{boettcher2022crystallography}. Throughout this paper, we will only be focusing on the two families:

\begin{equation}
\begin{aligned}
 \{4g, 4g\}&: (F_m, E_m, V_m) = (1, 2g, 1), \\
 \{2(2g+1), 2g+1\}&: (F_m, E_m, V_m) = (1, 2g+1, 2),
\end{aligned}
\end{equation}
where $g$ denotes the genus of the underlying Riemann surface. In particular, we focus on two patterns: $\{8,3\}$ and $\{10,3\}$. The Bravais lattice of the $\{8,3\}$ pattern, namely $\{8,8\}$, belongs to the $\{4g, 4g\}$ family for $g = 2$. Meanwhile, the Bravais lattices of the, $\{10,3\}$ pattern, namely $\{10,5\}$ belongs to the $\{2(2g+1),2g+1\}$ family for $g=2$. 

The closed Riemann surface obtained by imposing PBCs on the fundamental domain of the $\{8,8\}$ Bravais lattice was shown to possess the largest systole among all compact Riemann surfaces of genus $g = 2$ \cite{schmutz1993reimann}. This property is of particular importance in the context of quantum error correction, as the systole—defined as the length of the shortest non-contractible closed geodesic—provides a lower bound on the length of non-trivial logical operators. Consequently, a larger systole implies a greater code distance, enhancing the code’s robustness against logical errors. This property of the $\{8,8\}$ Bravais lattice makes the $\{8,3\}$ lattice an attractive candidate for constructing high-performance HQECCs. The compactified unit cell of this lattice, which realizes the corresponding genus-2 Riemann surface, is illustrated in Fig.~\ref{fig:polygon-riemann}.

\subsection{Fuchsian Groups}

Isometries of the Poincar\'e disk $\mathbb{D}$ are maps that preserve the hyperbolic metric and, in particular, the hyperbolic distance. We will be concerned with orientation-preserving isometries: these maps are given by elements of the group $PSU(1,1)=SU(1,1)/\{\pm \mathbb{I} \}$. Elements of $PSU(1,1)$ are given by linear transformations of the form
\begin{equation}
    z \rightarrow gz = \frac{az +b}{b^*z + a^*}, 
\end{equation}
where $g \in PSU(1,1)$, $z \in \mathbb{C}$ and $|a|^2-|b|^2=1$.
The full isometry group of $\mathbb{D}$ is given by the semi-direct product $G \ltimes \mathbb{Z}_2$, where elements of $\mathbb{Z}_2$ are the identity and the orientation reversing map $z \rightarrow z^*$. In other words, any isometry of $\mathbb{D}$ can be represented by an orientation-preserving map that is either combined or not combined with the orientation-reversing map $z \rightarrow z^*$. Had we used the upper-half plane model of the hyperbolic space, the isometry group would have been $PSL(2,\mathbb{R})$. However, because the two models are equivalent, the two groups $PSL(2,\mathbb{R})$ and $PSU(1,1)$ are isomorphic, as expected.

Every closed Riemann surface $M$ can be expressed as a quotient $S^2/\Gamma$, $\mathbb{R}^2/\Gamma$, or $\mathbb{H}/\Gamma$, where $ S^2$, $\mathbb{R}^2 $, and $\mathbb{H}$ denote the sphere, the Euclidean plane, and the hyperbolic plane, respectively, and $\Gamma$ is a discrete subgroup of isometries acting properly discontinuously on each space. If the Euler characteristic of $\chi(M)=0$, then
$g=1$ and $M$ can be described as $\mathbb{R}^2/\Gamma$. If
$\chi(M)>0$, then $g=0$ and $M$ can be described as $S^2/\Gamma$.
Finally, if $\chi(M)<0$, then $g\ge2$ and $M$ can be described as $\mathbb{H}/\Gamma$. 
We are concerned with hyperbolic surfaces for which $\chi(M)<0$ and $g \geq 2$. In this case, $\Gamma$ is called a Fuchsian group. A Fuchsian group $\Gamma$ is a discrete subgroup of $PSU(1,1)$; the elements of a Fuchsian group are the transformations that preserve the hyperbolic distance and hence leave the lattice invariant.

The space group of a hyperbolic pattern  $\{p,q\}$ is given by
$$SG_{\{p,q\}}= \langle a,b,c\,|\,a^2 = b^2 = c^2 = (ab)^2 = (bc)^p =(ca)^q = \mathbb{I} \rangle.$$
Elements of this group are the set of words consisting of $\{a,b,c,a^{-1},b^{-1},c^{-1} \}$ whereas group multiplication is simply a concatenation of words.
This group contains orientation-reversing elements and is therefore not a subgroup of $\Gamma$. If one considers the quotient by the orientation-reversing elements, one gets a Fuchsian group, that is a subgroup of $PSU(1,1)$ given by
$$ \Gamma= SG_{\{p,q\}}/\mathbb{Z}_2.$$
The Fuchsian group $\Gamma$ has the presentation
\begin{equation}
    \Gamma= \langle A,B \ | A^p = B^q = (AB)^2 = \mathbb{I} \rangle,
\end{equation}
where $A=bc$ and $B=ca$. Geometrically, $A$ is a rotation through the center of a face by an angle $\alpha = 2 \pi /p$ while $B$ is a rotation through a vertex by an angle $\beta = 2\pi /q$. 

Elements of a Fuchsian group are classified as elliptic, parabolic or hyperbolic if their trace, in the two-dimensional matrix representation, is less than, equal to, or greater than 2 respectively. Elliptic elements have one fixed point; therefore, these elements represent rotations and we denote them by $R(\theta)$, where $\theta$ is the angle of rotation. On the other hand, hyperbolic elements have no fixed points; therefore, they are considered translations (or boost transformations). We denote them by $T(\eta)$ where $\eta$ is the translation parameter.

A Fuchsian translation group is a torsion-free Fuchsian group, that is no element is of the form $\gamma^n=\mathbb{I}$. It is then obvious that $\Gamma$ given by (18) is not a Fuchsian translation group. The Fuchsian translation groups associated with Bravais lattices of the form $\{4g,4g\}$ and $\{2(2g+1),2g+1)\}$ in a hyperbolic space have been given an explicit representation in \cite{boettcher2022crystallography}. It has been shown that, for the two families, only $2g$ generators are independent, and the group is given by the following representation:
\begin{equation}
    \Gamma = \langle \gamma_1,...,\gamma_{p_B/2} \ | X_{\{p_B,q_B\}}= \mathbb{I}\rangle,
\end{equation}
where $\gamma_1$ is a boost transformation. The other generators $\gamma_m$, where $m = 1,...,2g$, are obtained by conjugating $\gamma_1$ with a rotation by a multiple of $\alpha = 2 \pi / p$
\begin{equation}
    \gamma_m = R((m-1)\alpha) \gamma_1 \ R(-(m-1)\alpha).
\end{equation} 
Furthermore, the constraint $X_{\{p_B,q_B\}}$ is the same for both families and only depends on  the $2g$ independent generators
\begin{equation}
    X_{\{p_B,q_B\}} = \gamma_1 \gamma_2^{-1} ... \gamma_{2g-1} \gamma_{2g}^{-1} \gamma_1^{-1} \gamma_2...\gamma_{2g-1}^{-1} \gamma_{2g}.
\end{equation}

From our previous discussion, we can conclude that hyperbolic polygon has a two-fold characteristic, it is the fundamental domain of a symmetry group and it can be used to construct closed surfaces via edge-pairing identification. We close this section by providing the necessary and sufficient conditions for a $p$-gon to be embedded in a closed Riemann surface \cite{stillwell1995geometry}. A hyperbolic $p$-gon $\Lambda$ can be embedded in a closed Riemann surface $M$ if the former is the fundamental domain of an orientation-preserving isometry group of $\mathbb{D}$ and if it satisfies the side and angle conditions given as follows
\begin{enumerate}
    \item for each edge $e$ in $\Lambda$ there is a unique edge $e'$ such that $e' = \gamma(e)$ for $\gamma \in \Gamma$, where $\gamma$ are the side-pairing transformations,
    \item For each set of vertices identified as a result of the edge-pairing transformations, the sum of the angles has to be equal to $2 \pi$.
\end{enumerate}
\begin{thm}{(Poincaré)}
    A compact polygon $P$ satisfying the side and angle conditions is the fundamental domain of the group $\Gamma$ generated by the side-pairing transformations of $P$.
\end{thm}

\section{Finite Hyperbolic lattices}

In this section, we outline a method for constructing finite hyperbolic \(\{p,q\}\) lattices given the Fuchsian group of the underlying Bravais lattice. Let \(\{p,q\}\) be a hyperbolic lattice whose unit cell $U$ lies within the fundamental domain of the Bravais lattice \(\{p_B,q_B\}\). Denote by \(\Gamma\) the Fuchsian group associated with the Bravais lattice \(\{p_B,q_B\}\), defined by the group presentation given by (19).

A finite \(\{p,q\}\) lattice $\mathcal{L}$ can be constructed by replicating the unit cell $U$ \(N\) times, where each copy is generated by applying an element \(g \in \Gamma\) to $U$. To ensure that the resulting lattice remains connected, we adopt a hierarchical approach to constructing these copies. We begin by applying elements of \(\Gamma\) that consist of a single generator—specifically, the generators \(\gamma_1, \dots, \gamma_{2g}\) and their inverses. Next, we apply elements of the form \(\gamma_j \gamma_k\) (where \(\gamma_j^{-1} \neq \gamma_k\)), then proceed to elements with longer word representations, iterating this process until the desired lattice size is achieved.

This finite $\{p,q\}$ lattice can be embedded in a genus \(g \geq 2\) Riemann surface by imposing appropriate PBCs. Mathematically, imposing the PBCs corresponds to selecting a normal subgroup \(\Gamma_{PBC}\) of index \(N\) within the Fuchsian group \(\Gamma\). The quotient group \(\Gamma / \Gamma_{PBC}\) then represents a finitely generated residual translation group acting on the finite lattice.

To identify index-\(N\) normal subgroups of a finitely presented group, computational group theory provides several algorithms \cite{conder2005applications, firth2005algorithm}. In this work, we employ the freely available computational algebra system GAP \cite{GAP4}, utilizing the LINS package \cite{LINS0.9}, which implements the low-index normal subgroup algorithm described in \cite{dietze1974determining}. However, this implementation is computationally expensive, as it relies on the Todd-Coxeter coset enumeration procedure \cite{todd1936practical}. Two distinct sources of inefficiency arise. First, the enumeration becomes impractical for large subgroup index $N$ (e.g., $N>25$) due to the large number of subgroups enumerated. Second, Fuchsian groups associated with Bravais lattices whose unit cells are embedded in higher-genus Riemann surfaces $(g \geq 3)$ have presentations with many generators, which significantly increases the complexity of the enumeration even for moderate values of $N$. We emphasize that this limitation affects only the subgroup search stage; once a suitable subgroup is identified, the subsequent construction and simulation procedures of the framework apply without modification to arbitrary genus. In practice, Bravais lattices with unit cells embedded in higher-genus surfaces $(g\ge 3)$ typically correspond to tessellations with larger values of $p$ and $q$, leading to higher-weight parity-check operators and consequently less favorable error-correction performance.

To overcome these limitations, a more efficient algorithm was proposed in \cite{chen2024anderson}, which enables the computation of normal subgroups of large indices \(N\). This method exploits the fact that if \(H\) and \(K\) are normal subgroups of a group \(G\), then their intersection \(L = H \cap K\) is also a normal subgroup of \(G\). Moreover, if \(h\), \(k\), and \(l\) denote the indices of \(H\), \(K\), and \(L\) in \(G\), respectively, then  
\[
\operatorname{lcm}(k,h) \leq l \leq kh,
\]
where equality holds when \(k\) and \(h\) are coprime.

In practice, we first compute low-index normal subgroups using the LINS package and then take intersections of different subgroups to generate normal subgroups of higher indices. This approach significantly expands the set of accessible normal subgroups beyond what the LINS package alone can achieve, making it a powerful tool for constructing finite hyperbolic lattices embedded in Riemann surfaces.

Let \(\Gamma_{PBC}\) be a subgroup of index \(N\) in \(\Gamma\). Then, \(\Gamma\) has the following coset decomposition:  
\begin{equation}
    \Gamma = \Gamma_{PBC} T_1 \cup \Gamma_{PBC} T_2 \cup \dots \cup \Gamma_{PBC} T_N,
\end{equation}
where the union is disjoint, and  
\begin{equation}
    T = \{T_1, T_2, \dots, T_N\} \subset \Gamma
\end{equation}
is a set of coset representatives, with \(T_1\) being the identity element in \(\Gamma\). The set \(T\) is known as the \textit{right transversal} of \(\Gamma_{PBC}\) in \(\Gamma\). The finite $\{p,q\}$ lattice \(\mathcal{L}\) is then given by  
\begin{equation}
    \mathcal{L} = \bigcup_{j=1}^N T_j U.
\end{equation}
The choice of a transversal is not unique since an element \(T_j \in T\) can be multiplied by any element of \(\Gamma_{PBC}\) while still satisfying (22). However, a physically meaningful choice of a transversal is one that ensures \(\mathcal{L}\) forms a connected graph when nearest-neighbor vertices are linked \cite{maciejko2022automorphic}.

Topologically, imposing PBCs can be understood in the framework of covering theory. If \(\Gamma_{PBC}\) is a normal subgroup of \(\Gamma\) of finite index \(N\), then the quotient group \(\Gamma_N = \Gamma / \Gamma_{PBC}\) is a finitely presented group of order \(N\). Each quotient group \(\Gamma_N\) serves as the symmetry group of a finite \(\{p,q\}\) lattice with \(N\) faces of the Bravais lattice. 

The minimal representation of the infinite \(\{p_B, q_B\}\) lattice is the comapctified unit cell obtained as the quotient space \(M = \mathbb{D} / \Gamma\), which defines a Riemann surface of genus \(g \geq 2\). The Fuchsian group \(\Gamma\) is isomorphic to the fundamental group of the quotient space, i.e.,  
\[
\Gamma \cong \pi_1(M).
\]
Similarly, imposing PBCs on a finite lattice \(\mathcal{L}\) with \(N\) faces by selecting a symmetry group \(\Gamma_{PBC}\) results in a Riemann surface \(M_N\) of genus \(h\), where \(M_N\) is an \(N\)-sheeted cover of \(M\). In this case,  
\[
\Gamma_{PBC} \cong \pi_1(M_N),
\]
and the Euler characteristic of \(M_N\) is given by
\begin{equation}
\begin{aligned}
    \chi(M_N) &= N \chi(M), \\
    2 - 2h &= N (2 - 2g).
\end{aligned}
\end{equation}
Thus, we recover a well-known result in algebraic geometry, the Riemann-Hurwitz formula:  
\begin{equation}
    h = N(g-1) + 1.
\end{equation}

We conclude this section by presenting a systematic framework for constructing a periodic $\{p,q\}$ lattice with an underlying $\{p_B,q_B\}$ Bravais lattice, given a Fuchsian group $\Gamma$ of $\{p_B,q_B\}$ and a normal subgroup $\Gamma_{PBC}$ of index $N$. We outline this framework in Algorithm 1. This algorithm incorporates the procedure outlined in \cite{chen2024anderson, tummuru2024hyperbolic} as a subroutine to construct the adjacency matrix $A_{\{p,q\}}$ of the $\{p,q\}$ lattice. An illustrative example of applying this procedure to impose PBCs on a hyperbolic lattice is shown in Fig.~\ref{fig:tiling-transformation}.

\begin{algorithm}[]
    \caption{Construction of a Periodic Regular $\{p,q\}$ Lattice}
    \KwIn{
        \begin{itemize}
            \item The values $p$, $q$, $p_B$ and $q_B$ defining the $\{p, q\}$ lattice and the $\{p_B, q_B\}$ Bravais lattice.
            \item A normal subgroup $\Gamma_{PBC}$ of index $N$ of the Fuchsian group $\Gamma$ associated with $\{p_B,q_B\}$.
        \end{itemize}
    }
    \KwOut{
        The graph $G_{PBC}$ representing the $\{p,q\}$ lattice after imposing the PBCs.
    }

    \textbf{Step 1: Generate the Unit Cell}  
    \begin{itemize}
        \item Compute the positions of the unit cell vertices of the $\{p,q\}$ lattice given by (6).
    \end{itemize}

    \textbf{Step 2: Generate the Fuchsian Group Generators}  
    \begin{itemize}
        \item Construct the generators of the Fuchsian group defined by (19).
    \end{itemize}

    \textbf{Step 3: Generate the Planar Graph $G$}  
    \begin{itemize}
        \item Apply a set of elements $\{g_1, g_2, \dots, g_{N-1}\} \subset \Gamma$ to the unit cell vertices to create additional $N-1$ copies of the unit cell.
        \item Connect nearest neighbors to construct the planar graph $G$ representing the $\{p,q\}$ lattice prior to imposing the PBCs. 
        \item The resulting graph $G$ consists of bulk vertices, each having degree $q$, and edge vertices, each having a degree smaller than $q$.
    \end{itemize}

    \textbf{Step 4: Construct the Adjacency Matrix $A_{\{p,q\}}$}  
    \begin{itemize}
        \item Let $V$ be the adjacency matrix of the unit cell of the $\{p,q\}$ lattice and initialize $$A_{\{p,q\}} = \mathbb{I}_N \otimes V,$$ where $\mathbb{I}_N$ is the identity matrix of size $N \times N$.
        \item For each $j = 1, \dots, p_B$, define the intercell matrix $I_j$ to represent the edges connecting each unit cell to its neighbor in the direction of $\gamma_j$, where $\gamma_j$ belongs to the set of generators of $\Gamma$ and their inverses. The matrix $I_j$ specifies the intercell connectivity pattern induced by $\gamma_j$.
        \item Let $A_{\{p_B,q_B\}}$ be the adjacency matrix of the Bravais lattice.
        \item For each pair of neighboring faces $(A_{\{p_B,q_B\}})_{n,m} = 1$ connected by $\gamma_j$, update $A_{\{p,q\}}$:
        \[
        A_{\{p,q\}} \rightarrow A_{\{p,q\}} + U\otimes I_j
        \]
        where $U$ is an $N \times N$ matrix with elements all zeros except $U_{nm} = 1$.
    \end{itemize}

    \textbf{Step 5: Imposing the PBCs}  
    \begin{itemize}
        \item Augment the graph $G$ by adding the edges in $A_{\{p,q\}}$ that are not already present in $G$. These additional edges arise from imposing the PBCs.
        \item We denote the resulting graph $G_{PBC}$; this graph is embedded in a Riemann surface whose genus is given by equation (26).
    \end{itemize}

\end{algorithm}

\begin{figure*}[ht]
    \centering
    \begin{minipage}{0.45\linewidth}
        \centering
        \includegraphics[width=\linewidth]{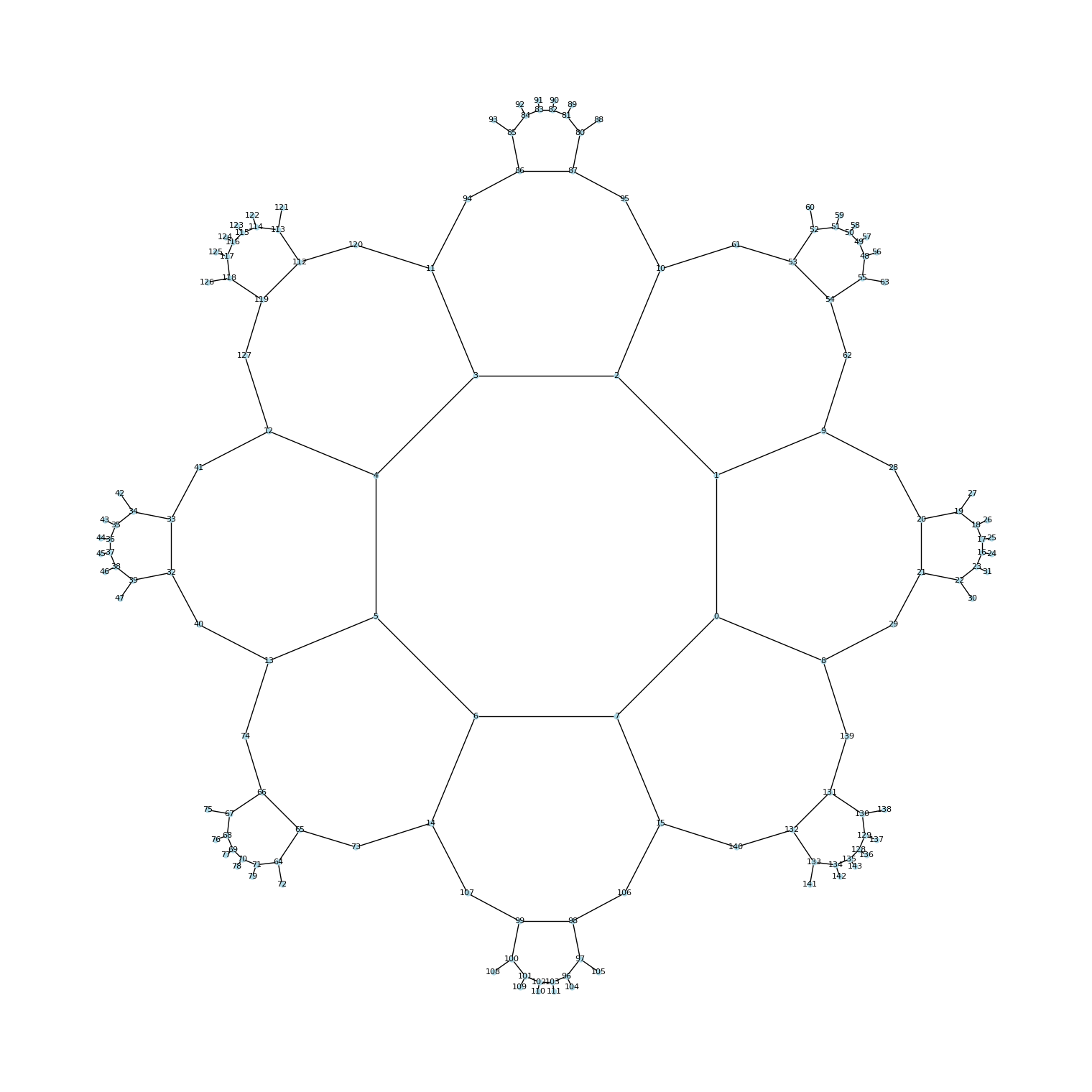}
    \end{minipage}
    \begin{minipage}{0.08\linewidth}
        \centering
        \begin{tikzpicture}
            \draw[thick,->] (0,0) -- (1.5,0);
        \end{tikzpicture}
    \end{minipage}
    \begin{minipage}{0.45\linewidth}
        \centering
        \includegraphics[width=\linewidth]{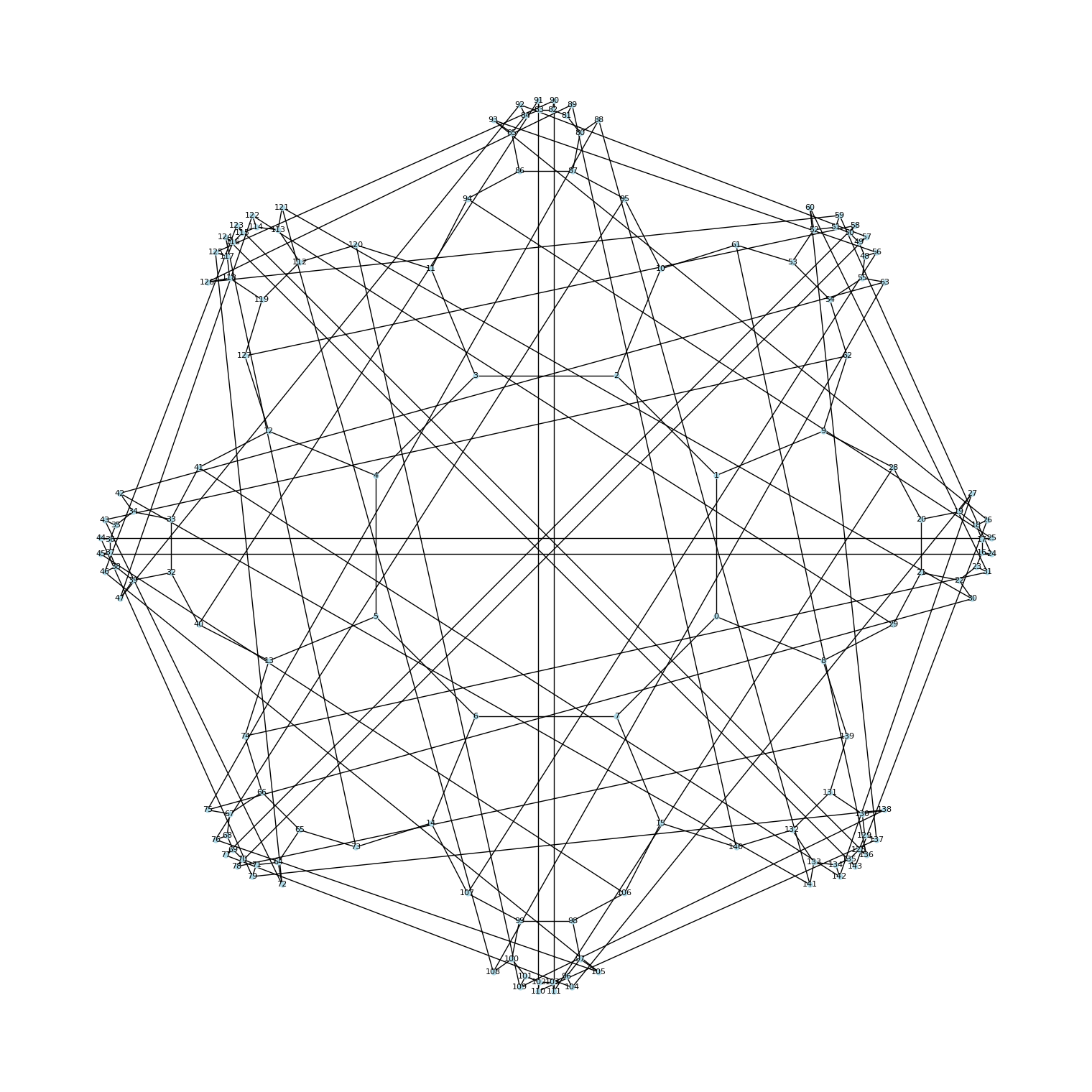}
    \end{minipage}

    \caption{The left half of the figure shows the finite $\{8,3\}$ lattice generated by $N=9$ faces of the Bravais lattice. The graph is generated by replicating the unit cell 8 times using the generators $\gamma_1,\dots,\gamma_4$ and their inverses. The right half of the figure shows the same lattice after imposing PBCs.}
    \label{fig:tiling-transformation}
\end{figure*}

\section{CSS Quantum Codes on Hyperbolic Lattices}
Topological quantum error correction codes introduced by Kitaev are a special class of stabilizer codes \cite{kitaev2003fault, freedman2003topological, gottesman1997stabilizer}. These codes are obtained by exploiting the topologies of closed surfaces and constructed from tessellations of these surfaces, which determines the placement of qubits and stabilizer operators. In the special case where the underlying geometry is Euclidean, one can tessellate a genus one torus by one of three patterns, $\{4,4\}$, $\{6,3\}$ or $\{3,6\}$. Embedding the qubits in a periodic $\{4,4\}$ pattern, one obtains the toric code, first introduced by Kitaev \cite{kitaev2003fault}. This construction generalizes naturally to hyperbolic geometry once the underlying closed Riemann surface has genus $g \geq 2$. In this section, we review the construction of a special class of topological quantum codes, the hyperbolic quantum error correction codes introduced in \cite{breuckmann2016constructions}. We start with a brief review of $\mathbb{Z}_2$ homology that is essential for the construction of HQECCs.

Let $M$ be a closed surface tessellated with a $\{p,q\}$ pattern. The faces, edges, and vertices of the tessellation define the $2$-, $1$-, and $0$-cells of a cellular decomposition of $M$, respectively. For each dimension $n \in \{0,1,2\}$, the $n$-chains, defined as $\mathbb{Z}_2$ linear combinations of the corresponding geometric elements, form a vector space over $\mathbb{Z}_2$ \cite{hatcher2002algebraic}. We denote this vector space by $C_n(M)$, or simply $C_n$ when the surface $M$ is clear from context. The boundary operators are defined by
$$\partial_n : C_n \rightarrow C_{n-1}. $$
That is, the boundary of an element $c_n \in C_n$ is the sum of all the $n-1$ cells incident on $c_n$. For example, let $v_1, v_2, v_3$ be the vertices of a triangle. The boundary of the edge $e_{12}$ connecting $v_1$ and $v_2$ is $\partial_1(e_{12})=v_2 + v_1$. Moreover, the boundary of the triangle $f_{123}$ is given by $\partial_2(f_{123})= e_3+e_2+e_1$, where  $e_1$, $e_2$, and $e_3$ are the three edges incident on that triangle.
The boundary operators satisfy the fundamental property
\begin{equation}
    \partial_{n-1} \circ \partial_n = 0,
\end{equation}
that is, the boundary of a boundary is zero.
The boundary operator $\partial_n$ defines two spaces, the cycle space $Z_1 = ker( \partial_2)$ and the boundary space $B_1 = im(\partial_{2})$. An obvious consequence of (27) is that $B_1 \subset Z_1$. The first homology group is defined as $H_1 = Z_1/B_1$. For an orientable surface, dimension of the first homology group is twice the genus $g$ of the surface, that is $\dim(H_1) = 2g$.

The dual map, the coboundary operator, is defined as follows
\begin{equation}
    \delta_n: C_n \rightarrow C_{n+1}.
\end{equation}
Working over $\mathbb{Z}_2$, we identify $C_n$ with its dual via the natural inner product 
$\langle c, c' \rangle = |c \cap c'| \pmod{2}$. Under this identification, the coboundary operator is the transpose of the boundary operator. Equivalently, for an $n$-cell, $\delta_n$ assigns the sum (modulo $2$) of all $(n+1)$-cells incident on it, and is extended linearly to arbitrary chains.

The coboundary operator defines the cocycle space 
$$Z^1 = \ker(\delta_1),$$
and the coboundary space 
$$B^1 = \operatorname{im}(\delta_0),$$
where $B^1 \subset Z^1$. The first cohomology group is defined as
$$H^1 = Z^1 / B^1.$$
With respect to the inner product above, one obtains the orthogonality relations
\begin{equation*}
    Z^1 = B_1^\perp,  
    \qquad 
    B^1 = Z_1^\perp.
\end{equation*}

Moreover, the groups $H_1(M)$ and $H^1(M)$ are related by the fact that $i$-cells of a tiled surface $M$ correspond to $(2-i)$-cells of the dual tiling $M^*$. In the present case $d=2$, and this correspondence induces an isomorphism
\begin{equation}
    * : C_i(M) \longrightarrow C_{2-i}(M^*).
\end{equation}
Under this identification, the coboundary operator on $M$ corresponds to the boundary operator on the dual tiling,
\begin{equation}
    \delta_i = *^{-1} \circ \partial_{2-i} \circ * .
\end{equation}

Now, we show how to utilize $\mathbb{Z}_2$ homology groups of a tiled surface to define topological quantum codes. A stabilizer code encoding $n$ physical qubits into $k$ logical qubits is a $2^k$ subspace $\mathcal{C} \subset \mathcal{H}$ of the Hilbert space of $n$ qubits $\mathcal{H} = (\mathbb{C}^2)^{\otimes n}$. The subspace $\mathcal{C} \subset \mathcal{H}$ is defined as the $+1$ eigenspace of an Abelian subgroup $S$ of the Pauli group. The transition from a tessellated surface $M$ to stabilizer codes is done by identifying all edges with qubits. The boundary of each face is identified with a $Z$ parity check operator while the coboundary of each vertex is identified with an $X$ parity check operator. This canonical set of operators are the generators of the stabilizer code. The weight of an operator is the number of qubits on which it acts non-trivially. Thus, for a $\{p,q\}$ tessellation of a closed surface $M$, the $Z$ parity check operators have weight $p$ while the $X$ parity check operators have weight $q$.

The number of physical qubits $n$ equals the number of edges, that is $n = \dim(C_1)$. On the other hand, the number of logical qubits is given by $n$ minus the constraints imposed by the stabilizer operators, that is
\begin{equation}
    \begin{aligned}
    k &= \dim(C_1)-\dim(B_1)-\dim(B^1), \\
    &= \dim(Z_1) - \dim(B_1), \\
    &= \dim(H_1),
    \end{aligned}
\end{equation}
where we used the fact that $$\dim(B^1)=\dim(Z_1^\perp)=\dim(C_1)-\dim(Z_1).$$

The distance of a stabilizer code is the minimum weight of a Pauli operator that commutes with all stabilizer generators but is not contained in the stabilizer group. In topological CSS surface codes with qubits placed on edges, the $Z$- and $X$-distances correspond to the lengths of the shortest homologically nontrivial cycles on the primal and dual tessellations, respectively; the code distance is $d=\min(d_X,d_Z)$.
The distance of a quantum CSS surface code can be computed by an algorithm due to Bravyi, as described in \cite{breuckmann2017hyperbolic}. For completeness, we outline this procedure in Algorithm 2.

\begin{algorithm}[H]
    \caption{Computation of the Distance $d_Z$ in a Topological CSS Code}
    \KwIn{
    \begin{itemize}
        \item Graph $G_{PBC} = (V, E)$ representing the finite hyperbolic lattice with PBCs.
        \item The set of logical operators $\{\bar{X}_1, \bar{X}_2, \dots, \bar{X}_k\}$.
    \end{itemize}
    }
    \KwOut{The distance $d_Z$ of the topological CSS code.}

    \For{$j = 1, \dots, k$}{
        Select a logical operator $\bar{X}_j$ and define its qubit support as $E(\bar{X}_j) \subset E$. Then, construct an auxiliary graph $\tilde{G}$ as follows:
        \begin{itemize}
            \item Create two copies of each vertex $v \in V$, labeled as $v^+$ and $v^-$.
            \item Initialize the edge set of $\tilde{G}$:
            \begin{itemize}
                \item For each edge $e = (u, v) \in E$:  
                \begin{itemize}
                    \item If $e \in E(\bar{X}_j)$, add edges $(u^+, v^-)$ and  $(u^-, v^+)$ to $\tilde{G}$.
                    \item If $e \notin E(\bar{X}_j)$, add edges $(u^+, v^+)$ and  $(u^-, v^-)$ to $\tilde{G}$.
                \end{itemize}
            \end{itemize}
        \end{itemize}

        Compute the shortest path distance $d(v^+, v^-)$ in $\tilde{G}$ for each vertex $v \in E(\bar{X}_j)$.
    }

    \textbf{Compute the code distance:}  
    Set $d_Z = \min d(v^+, v^-)$ over all $v \in E(\bar{X}_j)$ and all logical operators $\bar{X}_j$.

    \textbf{Computation of $d_X$:}  
    The same procedure can be applied to the set of logical operators $\{\bar{Z}_1, \bar{Z}_2, \dots, \bar{Z}_k\}$ to determine the distance $d_X$.
\label{alg:code_distance}
\end{algorithm}

\section{Hyperbolic Cycle Basis}

In this section, we present an algorithm for computing all plaquettes and logical operators in a HQECC. A plaquette is defined as a trivial cycle of length $p$, corresponding to a face of the finite $\{p, q\}$ hyperbolic lattice. In contrast, logical operators are associated with non-trivial cycles on the underlying Riemann surface $M$ of genus $g \geq 2$. While our method is inspired by the fundamental cycle basis algorithm of Paton \cite{paton1969algorithm} and the minimum cycle basis algorithm of Kavitha et al. \cite{kavitha2008algorithm}, it is fundamentally distinct. Rather than relying purely on graph-theoretic properties, our approach leverages the intrinsic geometry of the surface to construct a specific cycle basis that is essential for the implementation of HQECCs.

Let $G$ be a connected, finite, undirected graph with $E$ edges and $V$ vertices. A cycle $C$ in $G$ is any subgraph in which every vertex has even degree. Each cycle $C$ can be represented by an incidence vector $I_C \in \mathbb{Z}_2^{\otimes E}$, where $I_C^e = 1$ if and only if edge $e$ is part of the cycle $C$ (i.e., $e \in C$). The vector space generated by these incidence vectors is known as the cycle space of $G$.  A cycle basis of $G$ is a set of cycles that spans the cycle space. That is, given a cycle basis $\{C_1, C_2, ..., C_n\}$ of the graph $G$, any cycle $C$ in $G$ can then be expressed as:  
\begin{equation}
C = \sum_{j=1}^{n} \alpha_j C_j,
\end{equation}
where $\alpha_j \in \mathbb{Z}_2$.  

A \emph{fundamental cycle basis} (FCB) is a specific type of cycle basis for a graph \( G \). To construct an FCB, one first selects a spanning tree \( T(G) \) of the graph. For each edge \( e \in G \setminus T(G) \) that is not part of the spanning tree, a unique cycle is formed by combining \( e \) with the unique path in \( T(G) \) that connects the endpoints of \( e \). Each such cycle is called a fundamental cycle, and the collection of all such cycles constitutes the fundamental cycle basis.

Since a spanning tree on a graph with \( V \) vertices contains exactly \( V - 1 \) edges, the number of non-tree edges is \( E - V + 1 \). Therefore, the dimension of the cycle space spanned by a fundamental cycle basis is:
\[
\dim(\text{FCB}(G)) = E - V + 1.
\]

However, not every cycle basis is necessarily a fundamental cycle basis. Any set of linearly independent $E-V+1$ cycles that span the cycle space is a valid cycle basis of the graph. Since we are looking for a specific set of cycles, our choice is highly constrained by the hyperbolic geometry of the underlying lattice. To accommodate these constraints, the algorithm constructs a cycle basis that is not a fundamental cycle basis, but instead consists of cycles more naturally aligned with the geometry of the hyperbolic lattice.

Before presenting our algorithm for computing all plaquettes and logical operators of a given HQECC, we first prove that the set of plaquettes, excluding one plaquette, along with the set of logical operators forms a valid cycle basis for the underlying graph.

\vspace{12pt}  
\begin{thm}
Let $G_{PBC}$ be a graph representing a finite hyperbolic $\{p,q\}$ lattice embedded in a closed Riemann surface $M$ of genus $g \geq 2$. Suppose that $G_{PBC}$ has $V$ vertices, $E$ edges, and $F$ faces. Then a valid cycle basis for $G_{PBC}$ is given by the union of:
\begin{itemize}
    \item $F-1$ contractible cycles of length $p$, each corresponding to a plaquette of the hyperbolic lattice, and
    \item $2g$ non-contractible cycles each of length $l>p$ that generate the first homology group $H_1(M)$.
\end{itemize}
We denote this cycle basis as the \emph{Hyperbolic Cycle Basis (HCB)}.
\end{thm}

\begin{proof}

For the first part of the proof, we proceed by contradiction. Suppose that HCB is not a valid cycle basis for $G_{PBC}$, then there is a cycle $C \in Z_1$ that does not belong to the vector space spanned by the elements of HCB. Based on the discussion in Section IV , there are two types of cycles in the graph $G_{PBC}$, trivial cycles that are elements of $B_1$ and non-trivial cycles that are elements of $ Z_1 \setminus B_1$. The trivial cycles are either plaquettes or product of plaquettes. Since each edge is contained in exactly two plaquettes, the sum of all plaquette cycles (mod 2) is the identity. Hence, only a set of $F-1$ plaquette cycles are linearly independent, and the sum of these plaquette cycles equals the remaining plaquette cycle in the graph. Therefore, any set $\{Pl_1,Pl_2,...,Pl_{F-1}\}$ of $F-1$ plaquette cycles spans the space $B_1$. On the other hand, the vector space of non-trivial cycles is spanned by a set of $2g$ linearly independent cycles $\{h_1,h_2,...,h_{2g}\}$ that generate the first homology group of the underlying Riemann surface $H_1(M)$ and can not be expressed as a sum of plaquettes. To ensure linear independence, we require that these cycles are pairwise orthogonal. That is, $|h_j \cap h_k| \equiv 0 \pmod{2}$, for all $j,k \in \{1,...,2g\}$. 

Now, consider any cycle $C$ in the graph $G_{PBC}$. If $C \in B_1$, then $C$ can be expressed as 
\begin{equation}
C = \sum_{j=1}^{F-1} \alpha_j Pl_j,
\end{equation}
where $\alpha_j \in \mathbb{Z}_2$.
On the other hand, if $C \in Z_1 \setminus B_1$, then $C$ can be expressed as
\begin{equation}
C = \sum_{j=1}^{2g} \beta_j h_j,
\end{equation}
where $\beta \in \mathbb{Z}_2$. Thus, $C$ belongs to the vector space spanned by the elements of the $HCB$. In other words,
$$HCB = \{Pl_1,Pl_2,...,Pl_{F-1}, h_1,h_2,...,h_{2g}\}$$
is a valid cycle basis for $G_{PBC}$ 

To complete the proof, we need to show that the set  $HCB$ has the correct dimension of a cycle basis, that is 
\begin{equation}
    \dim(HCB)=E-V+1.
\end{equation} 
Combining (8) and (9),
\begin{equation}
\begin{aligned}
    2g &= 2 - \chi(M), \\
    &= 2- F + E - V.
\end{aligned}
\end{equation}
Therefore, the dimension of the Hyperbolic Cycle Basis is

\begin{equation*}
\begin{aligned}
    \dim(HCB) &= 2g + (F - 1), \\
              &= 2 - F + E - V + (F - 1), \\
              &= E-V+1. \qedhere
\end{aligned}
\end{equation*}
\end{proof}

Having established that the Hyperbolic Cycle Basis forms a valid cycle basis for the graph $G_{PBC}$, we now provide additional context before presenting Algorithm 3, which computes the HCB. The computational complexity of Algorithm 3 is analyzed separately in Appendix A.

The graph \( G \), constructed by replicating the unit cell of the underlying Bravais lattice, contains only trivial cycles—those belonging to the boundary subgroup \( B_1 \). Thus, any cycle of length \( p \) in \( G \) corresponds to a valid plaquette. Upon imposing PBCs, the resulting graph \( G_{PBC} \) contains \( 2E/p \) plaquettes, along with additional non-trivial cycles some of which might be of length $p$. Therefore, the plaquettes introduced as a result of imposing PBCs can be identified as those cycles of length $p$ that intersect all other plaquettes in at most one edge. Moreover, there are \(2g\) independent, mutually commuting Pauli-\(Z\) logical operators and \(2g\) mutually commuting Pauli-\(X\) logical operators, each associated with a distinct logical qubit. Importantly, these logical operators come in anti-commuting pairs, satisfying \( X_j Z_j = -Z_j X_j \) for all \( j = 1, \dots, 2g \). In this construction, Pauli-\(Z\) logical operators are represented by non-trivial cycles in the primal graph \(G_{PBC} \), whereas Pauli-\(X\) logical operators correspond to non-trivial cycles in the dual graph \( G^*_{PBC} \), which translate to non-trivial cocycles in \( G_{PBC} \). Identifying these non-trivial cycles requires more than selecting cycles of length greater than \( p \), as many such cycles are trivial, belong to $B_1$, and can be formed as products of plaquettes. Therefore, it is necessary to search specifically for cycles of length greater than \( p \) that can not be obtained as products of plaquettes.

Identifying higher-order trivial cycle can be done iteratively as follows. For every two plaquettes $F_1$ and $F_2$ that are neighbours, the higher-order trivial cycle is given by their symmetric difference. We call these cycles two-fold trivial cycles. The number of two-fold trivial cycles in a lattice equals the number of edges in this lattice and they all have length $2p-2$. Now consider a two-fold trivial cycle $C_p$. This cycle intersects some plaquettes in the lattice in at most two edges. For every plaquette intersecting this cycle, their product is a three-fold trivial cycle that has length $3p-2(i+1)$ where $i \in \{1,2\}$ is the length of the intersection between $C_p$ and the given plaquette. This process can be used iteratively to construct $n$-fold trivial cycles. An $n$-fold trivial cycle is given by the symmetric difference between an $(n-1)$-fold trivial cycle and an intersecting plaquette. Its length belongs to $np - 2(n-1) -2k$, where $k \in \{0,...,n-2\}$. In practice, it is typically sufficient to compute trivial cycles involving up to three-fold intersections, as these cycles already exhibit relatively large lengths, specifically $3p - 4$ and $3p - 6$. Such large lengths are generally adequate to identify all non-trivial cycles present in the lattice.

Fig.~\ref{fig:logicals} illustrates examples of logical operators obtained using the HCB algorithm for the \( \{8,3\} \) lattice with \( N = 9 \). It is important to note that, whereas every cycle necessarily includes at least one edge introduced by the imposition of PBCs, the same requirement does not apply to cocycles.

\begin{algorithm}[http]
\caption{Hyperbolic Cycle Basis Algorithm}
\KwIn{
\begin{itemize}
    \item Graph $G = (V, E')$ of the finite hyperbolic lattice before imposing PBCs.
    \item Graph $G_{PBC} = (V, E)$ of the lattice after imposing PBCs.
    \item Number of plaquettes $N_F = 2E/p$ in $G_{PBC}$.
    \item Genus of the embedding Riemann surface $g$.
\end{itemize}
}
\KwOut{\texttt{HCB},  \texttt{all\_plaquettes}, \texttt{Z\_logicals}, \texttt{X\_logicals}.}

\textbf{Step 1: Initial Plaquette Detection} \\
Find all cycles of length $p$ in $G$ and store them in two sets: \texttt{HCB} and \texttt{all\_plaquettes}. \\
Let $N_i = \texttt{len(HCB)}$.

\vspace{0.2em}
\textbf{Step 2: Add Remaining Plaquettes from $G_{PBC}$} \\
Find all cycles of length $p$ in $G_{PBC}$ that are not already in \texttt{HCB}; call this set \texttt{RP}.

\ForEach{$C_k \in \texttt{RP}$}{
    \If{$|C_k \cap C_p| \in \{0, 1\}$ for all $C_p \in \texttt{HCB}$}{
        \If{$\mathrm{len(HCB)} < N_F - 1$}{
            Append $C_k$ to \texttt{HCB}. \\
            Append $C_k$ to \texttt{all\_plaquettes}.
        }
        \ElseIf{$\mathrm{len(HCB)} == N_F - 1$}{
            Append $C_k$ to \texttt{all\_plaquettes}. \\
            \textbf{break}
        }
    }
}

\vspace{0.2em}
\textbf{Step 3: Generate Higher-Order Trivial Cycles} \\
Generate $n$-fold trivial cycles iteratively to obtain all trivial cycles up to a given maximum length $l$. \\
Store the output in the set \texttt{trivial\_cycles}.

\vspace{0.2em}
\textbf{Step 4: Find Non-Trivial Cycles ($Z$-type Logical Operators)} \\
Find all cycles $C_k$ of length $p < x \leq l$ in $G_{PBC}$ that are not in \texttt{trivial\_cycles}; denote this set as \texttt{potential\_logicals}.

\While{$\mathrm{len(Z\_logicals)} < 2 g$}{
    \ForEach{$Z_l \in \texttt{potential\_logicals}$}{
        \If{$|Z_l \cap Z_n| \equiv 0 \pmod{2}$ for all $Z_n \in \texttt{Z\_logicals}$}{
            Append $Z_l$ to \texttt{HCB}. \\
            Append $Z_l$ to \texttt{Z\_logicals}. \\
            \textbf{break}

        }
    }
}

\vspace{0.2em}
\textbf{Step 5: Find Non-Trivial Cocycles ($X$-type Logical Operators)} \\
{\begin{itemize}
    \item Construct the dual graph $G^*_{PBC}$, corresponding to a $\{q,p\}$ lattice, by placing a vertex at the center of each plaquette in $G_{PBC}$ and connecting two vertices whenever their corresponding plaquettes share a common edge.
    \item Apply steps 1-4 to the graph $G^*_{PBC}$ to find the set of non-trivial cycles in $G^*_{PBC}$. 
    \item Convert each non-trivial cycle \( C \) in \( G^*_{PBC} \) into its corresponding cocycle in \\ \( G_{PBC} \) and denote the resulting set of cocycles by \texttt{dual\_potentials}.
\end{itemize}}

\While{$\mathrm{len(X\_logicals)} < 2 g$}{ $k \gets len(\texttt{X\_logicals})$ \\
    \ForEach{$X_l \in \texttt{dual\_potentials}$}{
        \If{ $\left( |X_l \cap X_n| \equiv 0 \pmod{2} \quad \text{for all } X_n \in \texttt{X\_logicals} \right)$ \\
            \textbf{and} $\left( |X_l \cap Z_{k+1}| \equiv 1 \pmod{2} \right)$
           }{
            Append $X_l$ to \texttt{X\_logicals}. \\
            \textbf{break}
        }
    }
}

\end{algorithm}

\begin{figure*}[ht]
    \centering
    \begin{subfigure}{0.48\linewidth}
        \centering
        \includegraphics[width=\linewidth]{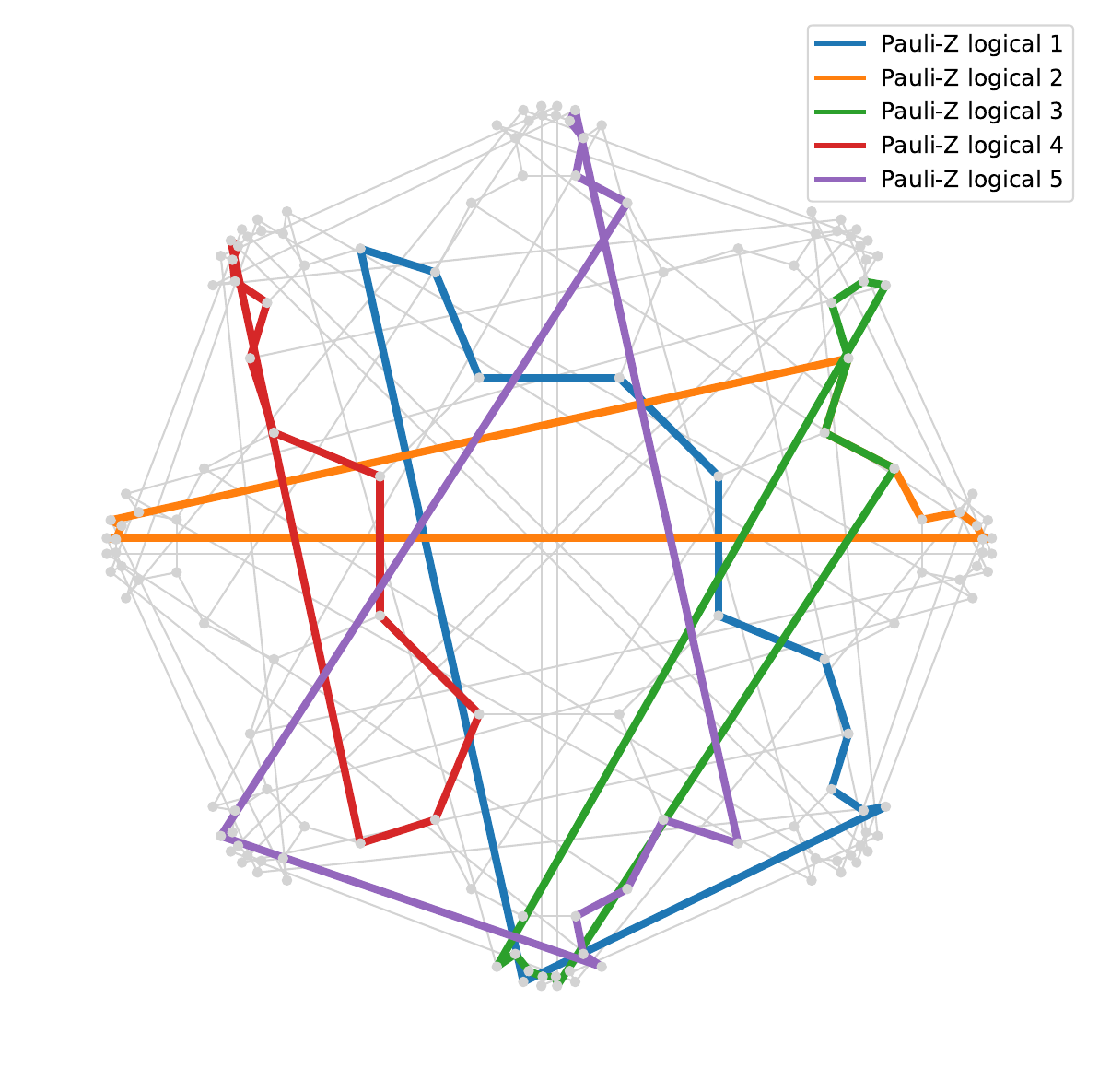}
        \caption{$Z$-type logical operators}
        \label{fig:z-logicals}
    \end{subfigure}
    \hfill
    \begin{subfigure}{0.48\linewidth}
        \centering
        \includegraphics[width=\linewidth]{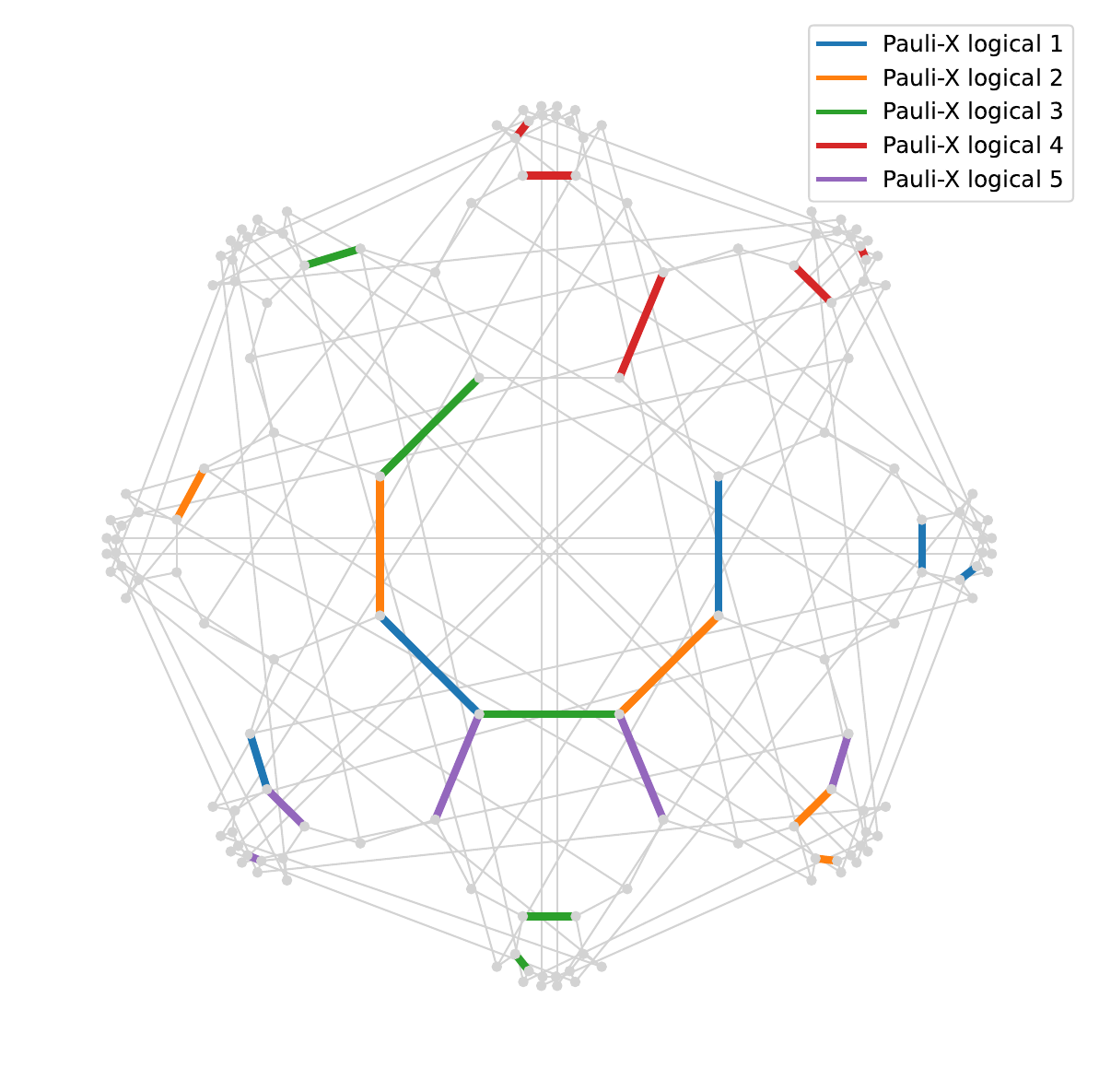}
        \caption{$X$-type logical operators}
        \label{fig:x-logicals}
    \end{subfigure}

    \caption{Logical operators are defined by canonical nontrivial cycles: $Z$-type operators correspond to nontrivial cycles on the primal graph, while $X$-type operators correspond to nontrivial cycles on the dual graph.}
    \label{fig:logicals}
\end{figure*}

Finally, we summarize our systematic framework for constructing HQECCs and computing various code parameters:

\begin{enumerate}
    \item Select a hyperbolic tessellation $\{p,q\}$ of the Poincar\'e disk with an associated Bravais lattice $\{p_B, q_B\}$.
    
    \item Apply Algorithm 1 to construct a periodic $\{p,q\}$ lattice. Let $G$ denote the graph representation of the $\{p,q\}$ lattice before imposing the PBCs, and let $G_{PBC}$ denote the corresponding graph after imposing the PBCs.
    
    \item Input $G$ and $G_{PBC}$ into Algorithm 3 to compute all plaquette cycles and determine the logical operators of the HQECC.
    
    \item Use $G_{PBC}$ obtained from Algorithm 1 and the set of logical operators obtained from Algorithm 3 as inputs into Algorithm 2 to compute the code distance $d$. 
    \item As discussed before, the number of physical qubits is the number of edges $E$ in $G_{PBC}$, and the number of logical qubits is given by $2h$, where $h$ is the genus of the underlying Riemann surface defined in (26). These are the parameters of the HQECC.

    \item Using this framework, one can also simulate a HQECC and make an estimate of the code's error threshold corresponding to a specified error model.
\end{enumerate}

\section{Numerical Simulations}

To demonstrate the utility and versatility of this framework, we apply it to simulate two representative HQECCs derived from two hyperbolic tessellations $\{8,3\}, \ \{10,3\}$ of the Poincar\'e disk, each associated with a distinct underlying Bravais lattice. The $\{8,3\}$ lattice admits an underlying $\{8,8\}$ Bravais lattice belonging to the $\{4g,4g\}$ family for $g=2$. On the other hand, the $\{10,3\}$ lattice has an underlying $\{10,5\}$ Bravais lattice belonging to the $\{2(2g+1),2g+1\}$ family for $g=2$.

For each simulation, we employ a phenomenological noise model in which each qubit experiences a Pauli-$Z$ error independently with probability $p$, and we assume that the syndrome measurements are error free. After that, we use a minimum-weight perfect matching (MWPM) decoder in order to infer the occurrence of a logical error. A logical error occurs if the product of the real and inferred errors anti-commutes with any of the $X$-type logical operators $\{X_1,...,X_{2g}\}$. The results of the simulations are depicted in Fig.~\ref{fig:threshold-combined}. An archived version of the simulation code used in this work is available at \cite{mahmoud2026hqecc_threshold}, while an actively maintained repository is hosted on GitHub at ~\cite{mahmoud_hqecc_github}.

For each HQECC, we compute the logical-error probability $p_L$ as a function of the physical error probability $p$ for several code sizes. To estimate the threshold, we examine the family of curves  and identify the parameter regime where the behavior with respect to system size inverts: for $p$ below the threshold, $p_L$ decreases with increasing code size, whereas for $p$ above the threshold, $p_L$ increases with increasing code size. The quoted threshold values $\approx3\%$ and $\approx 2\%$ for the $\{8,3\}$ and the $\{10,3\}$ HQECCs, respectively, are obtained from this criterion and should be understood as approximate estimates. These estimates are also in agreement with prior results showing that HQECCs thresholds range from $1-3\%$ when decoded with MWPM \cite{breuckmann2016constructions}. While faster decoders, such as union–find and belief propagation, have been proposed for surface codes, their decoding performance remains suboptimal compared with MWPM \cite{delfosse2021almost, huang2020fault, old2023generalized}.

The presented HQECCs possess a $Z/X$ asymmetry; therefore, the $Z$-error threshold differs from the $X$-error threshold. This contrasts with the toric code, which possesses $Z/X$ symmetry and exhibits an error threshold of approximately $10\%$ for the same noise model \cite{dennis2002topological}. 

The reduced threshold for the HQECCs compared with the toric code arises for two main reasons. First, the code distance of an HQECC grows only logarithmically with the number of physical qubits, $d=\Theta(\log n)$, in contrast to the toric code whose distance scales with the linear system size $d=\Theta(\sqrt{n})$ \cite{delfosse2013tradeoffs, breuckmann2016constructions}. This behavior reflects the tradeoff between distance and encoding rate: HQECCs achieve a constant encoding rate, whereas the toric code has vanishing rate in the thermodynamic limit. Consequently, the modest increase of the distance observed in Fig.~\ref{fig:threshold-combined} is not a peculiarity of the specific tessellations considered, but a generic feature of HQECCs, and this logarithmic distance scaling increases the logical failure probability compared with the toric code.

Second, hyperbolic tessellations correspond to closed surfaces of higher genus. Such higher-genus surfaces contain a large number of homologically nontrivial cycles, which increases the number of possible logical error paths and makes HQECCs more susceptible to logical failures. Together, these effects lead to lower thresholds for HQECCs compared with the toric code.

While our simulation present $Z$-error thresholds, the same framework can be applied to compute the $X$-error thresholds for both HQECCs. Generally, the framework naturally extends to circuit-level noise models with faulty syndrome measurements. 
In this case, repeated rounds of syndrome extraction stack the periodic hyperbolic lattice of Algorithm 1 into a space–time decoding graph, where the plaquette cycles and representatives of the logical operators derived from Algorithm 3 are defined on each temporal layer, and measurement faults introduce additional time-like edges between successive layers.

\begin{figure*}[t]
    \centering
    \includegraphics[width=0.49\linewidth]{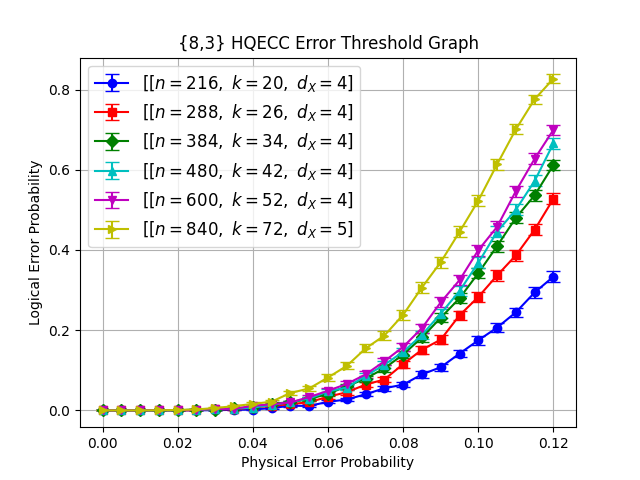}
    \includegraphics[width=0.49\linewidth]{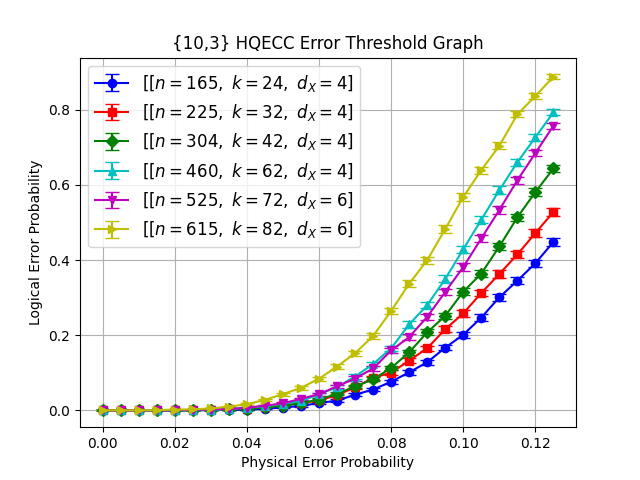}
    \caption{Estimated $Z$-error thresholds for two HQECCs with different tilings: $\{8,3\}$ with a threshold of approximately $3\%$ (left) and $\{10,3\}$ with a threshold of approximately $2\%$ (right). Each data point was obtained via $10^{4}$-trial Monte Carlo simulations under a phenomenological noise model, where each qubit experiences an independent Pauli-$Z$ error with probability $p$ per trial.}
    \label{fig:threshold-combined}
\end{figure*}

\section{Results}
In this work, we presented a systematic framework for constructing and benchmarking CSS quantum error correction codes on hyperbolic lattices. Central to this framework is the HCB algorithm, which efficiently computes all plaquette cycles and logical operators, represented by non-trivial cycles and cocycles on the hyperbolic lattice. The HCB algorithm enables a scalable and automated approach to defining the support of the code’s stabilizers and logical operators—an otherwise intractable task for large lattices if done manually. This capability is essential for benchmarking large-scale HQECCs. 

To demonstrate the capabilities of the framework, we applied it to compute the $Z$-error thresholds of two HQECCs constructed from the $\{8,3\}$ and $\{10,3\}$ hyperbolic tessellations. The simulations were performed using a phenomenological noise model in which each qubit independently experiences a Pauli-$Z$ error with probability $p$. Logical failures were identified when the error and its correction anti-commute with an $X$-type logical operator, and thresholds were estimated using Monte Carlo sampling over $10^{4}$ trials per data point.

Although we focused on HQECCs as representative examples, the framework can be adapted, with appropriate modifications, to benchmark other CSS codes defined on hyperbolic lattices. In particular, it extends naturally to codes with time‑dependent stabilizer measurements, such as CSS Floquet codes \cite{davydova2023floquet, kesselring2024anyon}. Floquet codes are defined by a periodic measurement schedule in which the instantaneous stabilizer group, and hence the chosen representatives of the logical operators, evolve in time. These codes can be implemented on 3‑colorable hyperbolic lattices. Algorithm 1 can be used to construct regular hyperbolic lattices of the form $\{p,3\}$ with PBCs. For suitable choices of $p$, these lattices are 3‑colorable, and therefore compatible with such Floquet measurement schedules \cite{soares2018hyperbolic}. Although the measured stabilizers and explicit logical representatives evolve during the Floquet cycle, they are associated with canonical plaquette cycles and nontrivial homology classes of the underlying lattice, respectively, which can be obtained from Algorithm 3 \cite{davydova2023floquet, higgott2024constructions}. Over one period, errors correspond to paths in the space–time decoding graph, and a logical failure occurs when an error path forms a nontrivial cycle that anti-commutes with a logical operator. Algorithm 2, which finds the shortest nontrivial spatial cycle with this property, thus provides a geometric baseline for estimating the code distance in the Floquet setting, with the precise distance depending on the full space–time decoding graph and the underlying noise model.

This work serves as a first step toward a systematic study of how lattice geometry influences the performance of CSS quantum error correction codes on hyperbolic lattices, a key question for optimizing future fault-tolerant quantum architectures. A companion study applying this framework to benchmark CSS Floquet codes on hyperbolic lattices will be presented elsewhere.

\begin{acknowledgments}
We thank the anonymous referees for their careful reading of the manuscript and for their constructive and valuable feedback, which significantly improved the clarity and quality of this work. We thank Joseph Maciejko for fruitful discussions around the use of GAP. We would also like to thank Srinivas Vamsi Parasa for helping in optimizing and parallelizing the Python code used to simulate the HQECCs. We acknowledge the use of Qiskit's AerSimulator for executing the HQECC circuits. SR has been partially supported by a Natural Sciences and Engineering Research Council of Canada (NSERC) Discovery Grant and a Canadian Tri-Agency New Frontiers in Research Fund (NFRF) Exploration Grant. AAM was supported during this work by the Discovery Grant of SR. 
\end{acknowledgments}

\appendix
\section{HCB Algorithm Complexity Analysis}

In this appendix we analyze the computational complexity of the HCB algorithm. Throughout, we consider the graph of a regular $\{p,q\}$ hyperbolic lattice. For such lattices, the number of vertices $V$, edges $E$, and faces $F$ scale linearly with one another according to (7), i.e.,
\begin{equation}
F = \Theta(E) = \Theta(V).
\end{equation}
All complexity estimates can, therefore, be expressed in terms of $E$. The HCB algorithm restricts attention to cycles of length at most $l$, where $l$ is an algorithmic cutoff parameter.

\subsection*{Step 1: Initial Plaquette Detection}

In Step 1, the algorithm identifies all cycles of length $p$ in the graph $G$ prior to imposing PBCs. Since $G$ admits a planar embedding without nontrivial cycles, every cycle of length $p$ corresponds to the boundary of a plaquette. Therefore, the number of length $p$ cycles in $G$ is $\mathcal{O}(F)$.

Enumerating all simple cycles in a graph is a classical problem, with Johnson's algorithm being its classical solution \cite{johnson1975finding}. For undirected graphs, an asymptotically optimal algorithm exists \cite{birmele2013optimal}. Here we use an implementation of Gupta's and Suzumura's algorithm to enumerate cycles of length up to $l$; this algorithm runs in $\mathcal{O}((\eta+n)(l-1)D^l)$, where $\eta$ is the number of cycles enumerated, $n$ is the number of vertices and $D$ is the average degree of the graph \cite{gupta2021finding}. For a $\{p,q\}$ regular hyperbolic lattice, $D=q$ and $n=V$. Moreover, in this step, $l=p$; therefore, Step 1 runs in
\begin{equation}
    T_1 = \mathcal{O}(F+V) = \mathcal{O}(E)
\end{equation}

\subsection*{Preprocessing: Detecting Length-$l$ Cycles}

For steps 2 and 4, we need to detect all cycles up to a given length $l$. Cycles of length $p$ will constitute the remaining plaquettes needed for Step 2, while higher-order cycles will be utilized to detect non-trivial cycles that represent $Z$-type logical operators. 
In practice, $l$ is initialized to a modest constant $l>p$ and increased iteratively until $2g$ independent nontrivial cycles are obtained in Step 4. The cost of finding all cycles up to a given length $l$ is
\begin{equation}
    T_{l-cycles} = \mathcal{O}((\eta+V)(l-1) q^l),
\end{equation}
where $\eta$ is the total number of cycles found.

\subsection*{Step 2: Adding Remaining Plaquettes from $G_{PBC}$}

In Step 2, the algorithm identifies any remaining plaquette cycles in $G_{\mathrm{PBC}}$ and augments the HCB set while enforcing the required intersection constraints. The number of candidate plaquettes is $\mathcal{O}(F)$, and the size of the HCB set grows to $\mathcal{O}(F)$.

The dominant cost in this step arises from checking pairwise intersections between candidate plaquettes and elements of the HCB set. Since each plaquette contains $p$ edges, each intersection test takes $\mathcal{O}(1)$ time. In the worst case, the nested loop performs $\mathcal{O}(F^2)$ such checks, yielding a total runtime
\begin{equation}
T_2 = \mathcal{O}(F^2) = \mathcal{O}(E^2).
\end{equation}

\subsection*{Step 3: Generating Higher-order Trivial Cycles}

In Step 3, the algorithm constructs higher-order trivial cycles by taking symmetric differences of plaquette boundaries. An $n$-fold trivial cycle corresponds to the boundary of a connected cluster of $n$ adjacent plaquettes. 

The construction of two-fold trivial cycles examines all face pairs. Since each face has constant size $p$, each intersection test costs $\mathcal{O}(1)$. Therefore, the two-fold stage runs in $\mathcal{O}(F^2)=\mathcal{O}(E^2)$.

For $n \geq 3$, the algorithm iterates over all previously generated $(n-1)$-fold cycles and checks each cycle with all $F$ plaquettes. If the $(n-1)$-fold cycle and the plaquette are adjacent, the symmetric difference between them is taken to generate an $n$-fold cycle. For each candidate pair, the symmetric difference between a cycle of size at most $l$ and a plaquette of constant size $p$ costs $\mathcal{O}(l)$. Consequently, the runtime for fixed $n$ is
\begin{equation}
    T_n = \mathcal{O}(F N_{n-1} l),
\end{equation}
where $N_{n-1}$ is the number of $(n-1)$-fold trivial cycles produced in the previous step. Summing over $n=3,...,n_{max}$, we get
\begin{equation}
    T_{n-cycles} = \mathcal{O}\left(\sum_{n=3}^{n_{max}} F N_{n-1}  l\right).
\end{equation}
Let $N_{tot}$ be the total number of trivial cycles generated, then
\begin{equation}
    T_3 = \mathcal{O}(E^2+EN_{tot}l).
\end{equation}

\subsection*{Step 4: Identifying $Z$-type Logical Operators}

In Step 4, the algorithm identifies nontrivial cycles representing $Z$-type logical operators by testing candidate cycles against previously selected logical operators using a parity-of-intersection criterion.

Let $\gamma$ denote the number of candidate cycles considered in this step. Each candidate cycle has length $\mathcal{O}(l)$, and the number of logical operators selected is $2g$. By (14), $g =cF+1$, where $c>0$ is a constant depending on $\{p,q\}$. Therefore, $g=\Theta(F)= \Theta(E)$. For each candidate, the algorithm checks its intersection parity with all previously selected logical operators. Each such check takes $\mathcal{O}(l)$ time; therefore, the total runtime of Step 4 is
\begin{equation}
T_4 = \mathcal{O}(\gamma g l) = \mathcal{O}(\gamma E l)
\end{equation}

\subsection*{Step 5: Identifying $X$-type Logical Operators}
Step 5 constructs $X$-type logical operators by working on the dual lattice. This step consists of two parts.

First, the algorithm constructs the dual graph $G^*_{\mathrm{PBC}}$ corresponding to the $\{q,p\}$ hyperbolic tiling. By (7), the number of edges in the primal and the dual graph is the same, while the number of faces and vertices are interchanged $F \leftrightarrow V$ in the two graphs.
We start by placing vertices in the center of each plaquette and then connect vertices corresponding to adjacent faces. The double loop to connect vertices takes $\mathcal{O}(F^2)$ time.

Second, Steps 1-4 of the HCB algorithm are applied to the dual graph $G^*_{\mathrm{PBC}}$ in order to identify nontrivial cycles in the dual lattice. Because $G^*_{\mathrm{PBC}}$ is again a bounded-degree hyperbolic tiling with constant face size and the number of edges in the primal and the dual graphs is the same, the same complexity bounds derived above apply. 

\subsection*{Summary}

The runtime complexity of the HCB algorithm is dominated by three components: the enumeration of cycles up to a prescribed cutoff length $l$ (A3) and the computation of intersections used to construct higher-order trivial cycles (A7) and identify logical operators (A8). These three components are executed separately for the primal and dual graphs.

To complement the asymptotic analysis, we measured the runtime required to generate the HCB graph for the $\{8,3\}$ tiling as a function of $E$. The results are shown in Fig.~\ref{fig:HCB_scaling}. For the system sizes accessible in this work $(E < 1000)$, the construction completes within a few seconds on a single CPU core. The figure serves to document the observed computational runtime for the instances considered here and is not intended to represent the asymptotic scaling behavior of the algorithm.

\begin{figure}[]
\centering
\includegraphics[width=0.9\columnwidth]{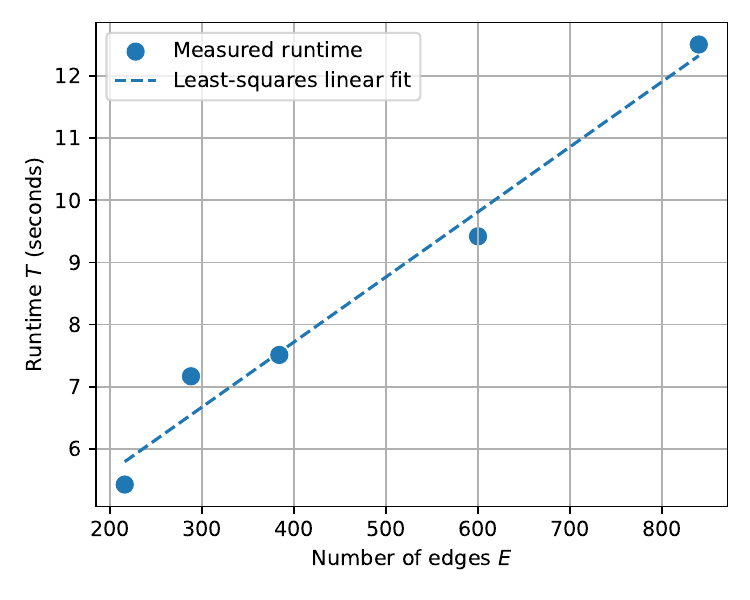}
\caption{
Measured runtime required to execute the HCB algorithm for five $\{8,3\}$ tilings considered in this work as a function of the number of edges $E$. The cycle–enumeration cutoff was set to $l=25$ for the primal graph and $l=10$ for the dual graph. The dashed line shows a least-squares linear fit to the data.
}
\label{fig:HCB_scaling}
\end{figure}

\bibliography{HQECCs}

\begin{thebibliography}{48}%
\makeatletter
\providecommand \@ifxundefined [1]{%
 \@ifx{#1\undefined}
}%
\providecommand \@ifnum [1]{%
 \ifnum #1\expandafter \@firstoftwo
 \else \expandafter \@secondoftwo
 \fi
}%
\providecommand \@ifx [1]{%
 \ifx #1\expandafter \@firstoftwo
 \else \expandafter \@secondoftwo
 \fi
}%
\providecommand \natexlab [1]{#1}%
\providecommand \enquote  [1]{``#1''}%
\providecommand \bibnamefont  [1]{#1}%
\providecommand \bibfnamefont [1]{#1}%
\providecommand \citenamefont [1]{#1}%
\providecommand \href@noop [0]{\@secondoftwo}%
\providecommand \href [0]{\begingroup \@sanitize@url \@href}%
\providecommand \@href[1]{\@@startlink{#1}\@@href}%
\providecommand \@@href[1]{\endgroup#1\@@endlink}%
\providecommand \@sanitize@url [0]{\catcode `\\12\catcode `\$12\catcode `\&12\catcode `\#12\catcode `\^12\catcode `\_12\catcode `\%12\relax}%
\providecommand \@@startlink[1]{}%
\providecommand \@@endlink[0]{}%
\providecommand \url  [0]{\begingroup\@sanitize@url \@url }%
\providecommand \@url [1]{\endgroup\@href {#1}{\urlprefix }}%
\providecommand \urlprefix  [0]{URL }%
\providecommand \Eprint [0]{\href }%
\providecommand \doibase [0]{https://doi.org/}%
\providecommand \selectlanguage [0]{\@gobble}%
\providecommand \bibinfo  [0]{\@secondoftwo}%
\providecommand \bibfield  [0]{\@secondoftwo}%
\providecommand \translation [1]{[#1]}%
\providecommand \BibitemOpen [0]{}%
\providecommand \bibitemStop [0]{}%
\providecommand \bibitemNoStop [0]{.\EOS\space}%
\providecommand \EOS [0]{\spacefactor3000\relax}%
\providecommand \BibitemShut  [1]{\csname bibitem#1\endcsname}%
\let\auto@bib@innerbib\@empty
\bibitem [{\citenamefont {Gottesman}(2009)}]{gottesman2009introduction}%
  \BibitemOpen
  \bibfield  {author} {\bibinfo {author} {\bibfnamefont {D.}~\bibnamefont {Gottesman}},\ }\bibfield  {title} {\bibinfo {title} {An introduction to quantum error correction and fault-tolerant quantum computation},\ }\href@noop {} {\bibfield  {journal} {\bibinfo  {journal} {arXiv preprint arXiv:0904.2557}\ } (\bibinfo {year} {2009})}\BibitemShut {NoStop}%
\bibitem [{\citenamefont {Terhal}(2015)}]{terhal2015quantum}%
  \BibitemOpen
  \bibfield  {author} {\bibinfo {author} {\bibfnamefont {B.~M.}\ \bibnamefont {Terhal}},\ }\bibfield  {title} {\bibinfo {title} {Quantum error correction for quantum memories},\ }\href@noop {} {\bibfield  {journal} {\bibinfo  {journal} {Reviews of Modern Physics}\ }\textbf {\bibinfo {volume} {87}},\ \bibinfo {pages} {307} (\bibinfo {year} {2015})}\BibitemShut {NoStop}%
\bibitem [{\citenamefont {Breuckmann}\ and\ \citenamefont {Eberhardt}(2021)}]{breuckmann2021quantum}%
  \BibitemOpen
  \bibfield  {author} {\bibinfo {author} {\bibfnamefont {N.~P.}\ \bibnamefont {Breuckmann}}\ and\ \bibinfo {author} {\bibfnamefont {J.~N.}\ \bibnamefont {Eberhardt}},\ }\bibfield  {title} {\bibinfo {title} {Quantum low-density parity-check codes},\ }\href@noop {} {\bibfield  {journal} {\bibinfo  {journal} {PRX quantum}\ }\textbf {\bibinfo {volume} {2}},\ \bibinfo {pages} {040101} (\bibinfo {year} {2021})}\BibitemShut {NoStop}%
\bibitem [{\citenamefont {Kitaev}(2003)}]{kitaev2003fault}%
  \BibitemOpen
  \bibfield  {author} {\bibinfo {author} {\bibfnamefont {A.~Y.}\ \bibnamefont {Kitaev}},\ }\bibfield  {title} {\bibinfo {title} {Fault-tolerant quantum computation by anyons},\ }\href@noop {} {\bibfield  {journal} {\bibinfo  {journal} {Annals of Physics}\ }\textbf {\bibinfo {volume} {303}},\ \bibinfo {pages} {2} (\bibinfo {year} {2003})}\BibitemShut {NoStop}%
\bibitem [{\citenamefont {Dennis}\ \emph {et~al.}(2002)\citenamefont {Dennis}, \citenamefont {Kitaev}, \citenamefont {Landahl},\ and\ \citenamefont {Preskill}}]{dennis2002topological}%
  \BibitemOpen
  \bibfield  {author} {\bibinfo {author} {\bibfnamefont {E.}~\bibnamefont {Dennis}}, \bibinfo {author} {\bibfnamefont {A.}~\bibnamefont {Kitaev}}, \bibinfo {author} {\bibfnamefont {A.}~\bibnamefont {Landahl}},\ and\ \bibinfo {author} {\bibfnamefont {J.}~\bibnamefont {Preskill}},\ }\bibfield  {title} {\bibinfo {title} {Topological quantum memory},\ }\href@noop {} {\bibfield  {journal} {\bibinfo  {journal} {Journal of Mathematical Physics}\ }\textbf {\bibinfo {volume} {43}},\ \bibinfo {pages} {4452} (\bibinfo {year} {2002})}\BibitemShut {NoStop}%
\bibitem [{\citenamefont {Breuckmann}\ and\ \citenamefont {Terhal}(2016)}]{breuckmann2016constructions}%
  \BibitemOpen
  \bibfield  {author} {\bibinfo {author} {\bibfnamefont {N.~P.}\ \bibnamefont {Breuckmann}}\ and\ \bibinfo {author} {\bibfnamefont {B.~M.}\ \bibnamefont {Terhal}},\ }\bibfield  {title} {\bibinfo {title} {Constructions and noise threshold of hyperbolic surface codes},\ }\href@noop {} {\bibfield  {journal} {\bibinfo  {journal} {IEEE Transactions on Information Theory}\ }\textbf {\bibinfo {volume} {62}},\ \bibinfo {pages} {3731} (\bibinfo {year} {2016})}\BibitemShut {NoStop}%
\bibitem [{\citenamefont {Breuckmann}\ \emph {et~al.}(2017)\citenamefont {Breuckmann}, \citenamefont {Vuillot}, \citenamefont {Campbell}, \citenamefont {Krishna},\ and\ \citenamefont {Terhal}}]{breuckmann2017hyperbolic}%
  \BibitemOpen
  \bibfield  {author} {\bibinfo {author} {\bibfnamefont {N.~P.}\ \bibnamefont {Breuckmann}}, \bibinfo {author} {\bibfnamefont {C.}~\bibnamefont {Vuillot}}, \bibinfo {author} {\bibfnamefont {E.}~\bibnamefont {Campbell}}, \bibinfo {author} {\bibfnamefont {A.}~\bibnamefont {Krishna}},\ and\ \bibinfo {author} {\bibfnamefont {B.~M.}\ \bibnamefont {Terhal}},\ }\bibfield  {title} {\bibinfo {title} {Hyperbolic and semi-hyperbolic surface codes for quantum storage},\ }\href@noop {} {\bibfield  {journal} {\bibinfo  {journal} {arXiv preprint arXiv:1703.00590}\ } (\bibinfo {year} {2017})}\BibitemShut {NoStop}%
\bibitem [{\citenamefont {Albuquerque}\ \emph {et~al.}(2009)\citenamefont {Albuquerque}, \citenamefont {Palazzo},\ and\ \citenamefont {Silva}}]{albuquerque2009topological}%
  \BibitemOpen
  \bibfield  {author} {\bibinfo {author} {\bibfnamefont {C.}~\bibnamefont {Albuquerque}}, \bibinfo {author} {\bibfnamefont {R.}~\bibnamefont {Palazzo}},\ and\ \bibinfo {author} {\bibfnamefont {E.}~\bibnamefont {Silva}},\ }\bibfield  {title} {\bibinfo {title} {Topological quantum codes on compact surfaces with genus $g\geq2$},\ }\href@noop {} {\bibfield  {journal} {\bibinfo  {journal} {Journal of mathematical physics}\ }\textbf {\bibinfo {volume} {50}} (\bibinfo {year} {2009})}\BibitemShut {NoStop}%
\bibitem [{\citenamefont {Kim}(2007)}]{kim2007quantum}%
  \BibitemOpen
  \bibfield  {author} {\bibinfo {author} {\bibfnamefont {I.~H.}\ \bibnamefont {Kim}},\ }\emph {\bibinfo {title} {Quantum codes on Hurwitz surfaces}},\ \href@noop {} {Ph.D. thesis},\ \bibinfo  {school} {Massachusetts Institute of Technology} (\bibinfo {year} {2007})\BibitemShut {NoStop}%
\bibitem [{\citenamefont {Hastings}\ and\ \citenamefont {Haah}(2021)}]{hastings2021dynamically}%
  \BibitemOpen
  \bibfield  {author} {\bibinfo {author} {\bibfnamefont {M.~B.}\ \bibnamefont {Hastings}}\ and\ \bibinfo {author} {\bibfnamefont {J.}~\bibnamefont {Haah}},\ }\bibfield  {title} {\bibinfo {title} {Dynamically generated logical qubits},\ }\href@noop {} {\bibfield  {journal} {\bibinfo  {journal} {Quantum}\ }\textbf {\bibinfo {volume} {5}},\ \bibinfo {pages} {564} (\bibinfo {year} {2021})}\BibitemShut {NoStop}%
\bibitem [{\citenamefont {Gidney}\ \emph {et~al.}(2021)\citenamefont {Gidney}, \citenamefont {Newman}, \citenamefont {Fowler},\ and\ \citenamefont {Broughton}}]{gidney2021honeycomb}%
  \BibitemOpen
  \bibfield  {author} {\bibinfo {author} {\bibfnamefont {C.}~\bibnamefont {Gidney}}, \bibinfo {author} {\bibfnamefont {M.}~\bibnamefont {Newman}}, \bibinfo {author} {\bibfnamefont {A.}~\bibnamefont {Fowler}},\ and\ \bibinfo {author} {\bibfnamefont {M.}~\bibnamefont {Broughton}},\ }\bibfield  {title} {\bibinfo {title} {A fault-tolerant honeycomb memory},\ }\href@noop {} {\bibfield  {journal} {\bibinfo  {journal} {Quantum}\ }\textbf {\bibinfo {volume} {5}},\ \bibinfo {pages} {605} (\bibinfo {year} {2021})}\BibitemShut {NoStop}%
\bibitem [{\citenamefont {Higgott}\ and\ \citenamefont {Breuckmann}(2024)}]{higgott2024constructions}%
  \BibitemOpen
  \bibfield  {author} {\bibinfo {author} {\bibfnamefont {O.}~\bibnamefont {Higgott}}\ and\ \bibinfo {author} {\bibfnamefont {N.~P.}\ \bibnamefont {Breuckmann}},\ }\bibfield  {title} {\bibinfo {title} {Constructions and performance of hyperbolic and semi-hyperbolic floquet codes},\ }\href@noop {} {\bibfield  {journal} {\bibinfo  {journal} {PRX Quantum}\ }\textbf {\bibinfo {volume} {5}},\ \bibinfo {pages} {040327} (\bibinfo {year} {2024})}\BibitemShut {NoStop}%
\bibitem [{\citenamefont {Ratcliffe}(2006)}]{ratcliffe2006foundations}%
  \BibitemOpen
  \bibfield  {author} {\bibinfo {author} {\bibfnamefont {J.~G.}\ \bibnamefont {Ratcliffe}},\ }\href@noop {} {\emph {\bibinfo {title} {Foundations of hyperbolic manifolds}}}\ (\bibinfo  {publisher} {Springer},\ \bibinfo {year} {2006})\BibitemShut {NoStop}%
\bibitem [{\citenamefont {Boettcher}\ \emph {et~al.}(2022)\citenamefont {Boettcher}, \citenamefont {Gorshkov}, \citenamefont {Koll{\'a}r}, \citenamefont {Maciejko}, \citenamefont {Rayan},\ and\ \citenamefont {Thomale}}]{boettcher2022crystallography}%
  \BibitemOpen
  \bibfield  {author} {\bibinfo {author} {\bibfnamefont {I.}~\bibnamefont {Boettcher}}, \bibinfo {author} {\bibfnamefont {A.~V.}\ \bibnamefont {Gorshkov}}, \bibinfo {author} {\bibfnamefont {A.~J.}\ \bibnamefont {Koll{\'a}r}}, \bibinfo {author} {\bibfnamefont {J.}~\bibnamefont {Maciejko}}, \bibinfo {author} {\bibfnamefont {S.}~\bibnamefont {Rayan}},\ and\ \bibinfo {author} {\bibfnamefont {R.}~\bibnamefont {Thomale}},\ }\bibfield  {title} {\bibinfo {title} {Crystallography of hyperbolic lattices},\ }\href@noop {} {\bibfield  {journal} {\bibinfo  {journal} {Physical Review B}\ }\textbf {\bibinfo {volume} {105}},\ \bibinfo {pages} {125118} (\bibinfo {year} {2022})}\BibitemShut {NoStop}%
\bibitem [{\citenamefont {Koll{\'a}r}\ \emph {et~al.}(2019)\citenamefont {Koll{\'a}r}, \citenamefont {Fitzpatrick},\ and\ \citenamefont {Houck}}]{kollar2019hyperbolic}%
  \BibitemOpen
  \bibfield  {author} {\bibinfo {author} {\bibfnamefont {A.~J.}\ \bibnamefont {Koll{\'a}r}}, \bibinfo {author} {\bibfnamefont {M.}~\bibnamefont {Fitzpatrick}},\ and\ \bibinfo {author} {\bibfnamefont {A.~A.}\ \bibnamefont {Houck}},\ }\bibfield  {title} {\bibinfo {title} {Hyperbolic lattices in circuit quantum electrodynamics},\ }\href@noop {} {\bibfield  {journal} {\bibinfo  {journal} {Nature}\ }\textbf {\bibinfo {volume} {571}},\ \bibinfo {pages} {45} (\bibinfo {year} {2019})}\BibitemShut {NoStop}%
\bibitem [{\citenamefont {Bzdu{\v{s}}ek}\ and\ \citenamefont {Maciejko}(2022)}]{Bzdusek2022}%
  \BibitemOpen
  \bibfield  {author} {\bibinfo {author} {\bibfnamefont {T.}~\bibnamefont {Bzdu{\v{s}}ek}}\ and\ \bibinfo {author} {\bibfnamefont {J.}~\bibnamefont {Maciejko}},\ }\bibfield  {title} {\bibinfo {title} {Flat bands and band-touching from real-space topology in hyperbolic lattices},\ }\href@noop {} {\bibfield  {journal} {\bibinfo  {journal} {Physical Review B}\ }\textbf {\bibinfo {volume} {106}},\ \bibinfo {pages} {155146} (\bibinfo {year} {2022})}\BibitemShut {NoStop}%
\bibitem [{\citenamefont {Koll{\'a}r}\ \emph {et~al.}(2020)\citenamefont {Koll{\'a}r}, \citenamefont {Fitzpatrick}, \citenamefont {Sarnak},\ and\ \citenamefont {Houck}}]{Kollar2020}%
  \BibitemOpen
  \bibfield  {author} {\bibinfo {author} {\bibfnamefont {A.~J.}\ \bibnamefont {Koll{\'a}r}}, \bibinfo {author} {\bibfnamefont {M.}~\bibnamefont {Fitzpatrick}}, \bibinfo {author} {\bibfnamefont {P.}~\bibnamefont {Sarnak}},\ and\ \bibinfo {author} {\bibfnamefont {A.~A.}\ \bibnamefont {Houck}},\ }\bibfield  {title} {\bibinfo {title} {Line-graph lattices: Euclidean and non-euclidean flat bands, and implementations in circuit quantum electrodynamics},\ }\href@noop {} {\bibfield  {journal} {\bibinfo  {journal} {Communications in Mathematical Physics}\ }\textbf {\bibinfo {volume} {376}},\ \bibinfo {pages} {1909} (\bibinfo {year} {2020})}\BibitemShut {NoStop}%
\bibitem [{\citenamefont {Beardon}(2012)}]{beardon2012geometry}%
  \BibitemOpen
  \bibfield  {author} {\bibinfo {author} {\bibfnamefont {A.~F.}\ \bibnamefont {Beardon}},\ }\href@noop {} {\emph {\bibinfo {title} {The Geometry of Discrete Groups}}},\ Vol.~\bibinfo {volume} {91}\ (\bibinfo  {publisher} {Springer},\ \bibinfo {year} {2012})\BibitemShut {NoStop}%
\bibitem [{\citenamefont {Katok}(1992)}]{katok1992fuchsian}%
  \BibitemOpen
  \bibfield  {author} {\bibinfo {author} {\bibfnamefont {S.}~\bibnamefont {Katok}},\ }\href@noop {} {\emph {\bibinfo {title} {Fuchsian Groups}}}\ (\bibinfo  {publisher} {University of Chicago Press},\ \bibinfo {year} {1992})\BibitemShut {NoStop}%
\bibitem [{\citenamefont {Stillwell}(1995)}]{stillwell1995geometry}%
  \BibitemOpen
  \bibfield  {author} {\bibinfo {author} {\bibfnamefont {J.}~\bibnamefont {Stillwell}},\ }\href@noop {} {\emph {\bibinfo {title} {Geometry of Surfaces}}}\ (\bibinfo  {publisher} {Springer},\ \bibinfo {year} {1995})\BibitemShut {NoStop}%
\bibitem [{\citenamefont {Do~Carmo}(2016)}]{do2016differential}%
  \BibitemOpen
  \bibfield  {author} {\bibinfo {author} {\bibfnamefont {M.~P.}\ \bibnamefont {Do~Carmo}},\ }\href@noop {} {\emph {\bibinfo {title} {Differential Geometry of Curves and Surfaces: Revised and Updated Second Edition}}}\ (\bibinfo  {publisher} {Courier Dover Publications},\ \bibinfo {year} {2016})\BibitemShut {NoStop}%
\bibitem [{\citenamefont {Schmutz}(1993)}]{schmutz1993reimann}%
  \BibitemOpen
  \bibfield  {author} {\bibinfo {author} {\bibfnamefont {P.}~\bibnamefont {Schmutz}},\ }\bibfield  {title} {\bibinfo {title} {Riemann surfaces with shortest geodesic of maximal length},\ }\href@noop {} {\bibfield  {journal} {\bibinfo  {journal} {Geometric and Functional Analysis (GAFA)}\ }\textbf {\bibinfo {volume} {3}},\ \bibinfo {pages} {564} (\bibinfo {year} {1993})}\BibitemShut {NoStop}%
\bibitem [{\citenamefont {Conder}\ and\ \citenamefont {Dobcs{\'a}nyi}(2005)}]{conder2005applications}%
  \BibitemOpen
  \bibfield  {author} {\bibinfo {author} {\bibfnamefont {M.}~\bibnamefont {Conder}}\ and\ \bibinfo {author} {\bibfnamefont {P.}~\bibnamefont {Dobcs{\'a}nyi}},\ }\bibfield  {title} {\bibinfo {title} {Applications and adaptations of the low index subgroups procedure},\ }\href@noop {} {\bibfield  {journal} {\bibinfo  {journal} {Mathematics of Computation}\ }\textbf {\bibinfo {volume} {74}},\ \bibinfo {pages} {485} (\bibinfo {year} {2005})}\BibitemShut {NoStop}%
\bibitem [{\citenamefont {Firth}(2005)}]{firth2005algorithm}%
  \BibitemOpen
  \bibfield  {author} {\bibinfo {author} {\bibfnamefont {D.}~\bibnamefont {Firth}},\ }\emph {\bibinfo {title} {An algorithm to find normal subgroups of a finitely presented group, up to a given finite index}},\ \href@noop {} {Ph.D. thesis},\ \bibinfo  {school} {University of Warwick} (\bibinfo {year} {2005})\BibitemShut {NoStop}%
\bibitem [{GAP(2021)}]{GAP4}%
  \BibitemOpen
  \href {https://www.gap-system.org} {\emph {\bibinfo {title} {GAP -- Groups, Algorithms, and Programming, Version 4.11.1}}},\ \bibinfo {organization} {The GAP Group} (\bibinfo {year} {2021})\BibitemShut {NoStop}%
\bibitem [{\citenamefont {Rober}(2024)}]{LINS0.9}%
  \BibitemOpen
  \bibfield  {author} {\bibinfo {author} {\bibfnamefont {F.}~\bibnamefont {Rober}},\ }\href@noop {} {\bibinfo {title} {Lins: provides an algorithm for computing the normal subgroups of a finitely presented group up to some given index bound, version 0.9}},\ \bibinfo {howpublished} {\\url{https://gap-packages.github.io/LINS/}} (\bibinfo {year} {2024}),\ \bibinfo {note} {gAP package}\BibitemShut {NoStop}%
\bibitem [{\citenamefont {Dietze}\ and\ \citenamefont {Schaps}(1974)}]{dietze1974determining}%
  \BibitemOpen
  \bibfield  {author} {\bibinfo {author} {\bibfnamefont {A.}~\bibnamefont {Dietze}}\ and\ \bibinfo {author} {\bibfnamefont {M.}~\bibnamefont {Schaps}},\ }\bibfield  {title} {\bibinfo {title} {Determining subgroups of a given finite index in a finitely presented group},\ }\href@noop {} {\bibfield  {journal} {\bibinfo  {journal} {Canadian Journal of Mathematics}\ }\textbf {\bibinfo {volume} {26}},\ \bibinfo {pages} {769} (\bibinfo {year} {1974})}\BibitemShut {NoStop}%
\bibitem [{\citenamefont {Todd}\ and\ \citenamefont {Coxeter}(1936)}]{todd1936practical}%
  \BibitemOpen
  \bibfield  {author} {\bibinfo {author} {\bibfnamefont {J.~A.}\ \bibnamefont {Todd}}\ and\ \bibinfo {author} {\bibfnamefont {H.~S.~M.}\ \bibnamefont {Coxeter}},\ }\bibfield  {title} {\bibinfo {title} {A practical method for enumerating cosets of a finite abstract group},\ }\href@noop {} {\bibfield  {journal} {\bibinfo  {journal} {Proceedings of the Edinburgh Mathematical Society}\ }\textbf {\bibinfo {volume} {5}},\ \bibinfo {pages} {26} (\bibinfo {year} {1936})}\BibitemShut {NoStop}%
\bibitem [{\citenamefont {Chen}\ \emph {et~al.}(2024)\citenamefont {Chen}, \citenamefont {Maciejko},\ and\ \citenamefont {Boettcher}}]{chen2024anderson}%
  \BibitemOpen
  \bibfield  {author} {\bibinfo {author} {\bibfnamefont {A.}~\bibnamefont {Chen}}, \bibinfo {author} {\bibfnamefont {J.}~\bibnamefont {Maciejko}},\ and\ \bibinfo {author} {\bibfnamefont {I.}~\bibnamefont {Boettcher}},\ }\bibfield  {title} {\bibinfo {title} {Anderson localization transition in disordered hyperbolic lattices},\ }\href@noop {} {\bibfield  {journal} {\bibinfo  {journal} {Physical Review Letters}\ }\textbf {\bibinfo {volume} {133}},\ \bibinfo {pages} {066101} (\bibinfo {year} {2024})}\BibitemShut {NoStop}%
\bibitem [{\citenamefont {Maciejko}\ and\ \citenamefont {Rayan}(2022)}]{maciejko2022automorphic}%
  \BibitemOpen
  \bibfield  {author} {\bibinfo {author} {\bibfnamefont {J.}~\bibnamefont {Maciejko}}\ and\ \bibinfo {author} {\bibfnamefont {S.}~\bibnamefont {Rayan}},\ }\bibfield  {title} {\bibinfo {title} {Automorphic bloch theorems for hyperbolic lattices},\ }\href@noop {} {\bibfield  {journal} {\bibinfo  {journal} {Proceedings of the National Academy of Sciences}\ }\textbf {\bibinfo {volume} {119}},\ \bibinfo {pages} {e2116869119} (\bibinfo {year} {2022})}\BibitemShut {NoStop}%
\bibitem [{\citenamefont {Tummuru}\ \emph {et~al.}(2024)\citenamefont {Tummuru}, \citenamefont {Chen}, \citenamefont {Lenggenhager}, \citenamefont {Neupert}, \citenamefont {Maciejko},\ and\ \citenamefont {Bzdu{\v{s}}ek}}]{tummuru2024hyperbolic}%
  \BibitemOpen
  \bibfield  {author} {\bibinfo {author} {\bibfnamefont {T.}~\bibnamefont {Tummuru}}, \bibinfo {author} {\bibfnamefont {A.}~\bibnamefont {Chen}}, \bibinfo {author} {\bibfnamefont {P.~M.}\ \bibnamefont {Lenggenhager}}, \bibinfo {author} {\bibfnamefont {T.}~\bibnamefont {Neupert}}, \bibinfo {author} {\bibfnamefont {J.}~\bibnamefont {Maciejko}},\ and\ \bibinfo {author} {\bibfnamefont {T.}~\bibnamefont {Bzdu{\v{s}}ek}},\ }\bibfield  {title} {\bibinfo {title} {Hyperbolic non-abelian semimetal},\ }\href@noop {} {\bibfield  {journal} {\bibinfo  {journal} {Physical Review Letters}\ }\textbf {\bibinfo {volume} {132}},\ \bibinfo {pages} {206601} (\bibinfo {year} {2024})}\BibitemShut {NoStop}%
\bibitem [{\citenamefont {Freedman}\ \emph {et~al.}(2003)\citenamefont {Freedman}, \citenamefont {Kitaev}, \citenamefont {Larsen},\ and\ \citenamefont {Wang}}]{freedman2003topological}%
  \BibitemOpen
  \bibfield  {author} {\bibinfo {author} {\bibfnamefont {M.}~\bibnamefont {Freedman}}, \bibinfo {author} {\bibfnamefont {A.}~\bibnamefont {Kitaev}}, \bibinfo {author} {\bibfnamefont {M.}~\bibnamefont {Larsen}},\ and\ \bibinfo {author} {\bibfnamefont {Z.}~\bibnamefont {Wang}},\ }\bibfield  {title} {\bibinfo {title} {Topological quantum computation},\ }\href@noop {} {\bibfield  {journal} {\bibinfo  {journal} {Bulletin of the American Mathematical Society}\ }\textbf {\bibinfo {volume} {40}},\ \bibinfo {pages} {31} (\bibinfo {year} {2003})}\BibitemShut {NoStop}%
\bibitem [{\citenamefont {Gottesman}(1997)}]{gottesman1997stabilizer}%
  \BibitemOpen
  \bibfield  {author} {\bibinfo {author} {\bibfnamefont {D.}~\bibnamefont {Gottesman}},\ }\emph {\bibinfo {title} {Stabilizer codes and quantum error correction}},\ \href@noop {} {Ph.D. thesis},\ \bibinfo  {school} {California Institute of Technology} (\bibinfo {year} {1997})\BibitemShut {NoStop}%
\bibitem [{\citenamefont {Hatcher}(2002)}]{hatcher2002algebraic}%
  \BibitemOpen
  \bibfield  {author} {\bibinfo {author} {\bibfnamefont {A.}~\bibnamefont {Hatcher}},\ }\href@noop {} {\emph {\bibinfo {title} {Algebraic Topology}}}\ (\bibinfo  {publisher} {Cambridge University Press},\ \bibinfo {year} {2002})\BibitemShut {NoStop}%
\bibitem [{\citenamefont {Paton}(1969)}]{paton1969algorithm}%
  \BibitemOpen
  \bibfield  {author} {\bibinfo {author} {\bibfnamefont {K.}~\bibnamefont {Paton}},\ }\bibfield  {title} {\bibinfo {title} {An algorithm for finding a fundamental set of cycles of a graph},\ }\href@noop {} {\bibfield  {journal} {\bibinfo  {journal} {Communications of the ACM}\ }\textbf {\bibinfo {volume} {12}},\ \bibinfo {pages} {514} (\bibinfo {year} {1969})}\BibitemShut {NoStop}%
\bibitem [{\citenamefont {Kavitha}\ \emph {et~al.}(2008)\citenamefont {Kavitha}, \citenamefont {Mehlhorn}, \citenamefont {Michail},\ and\ \citenamefont {Paluch}}]{kavitha2008algorithm}%
  \BibitemOpen
  \bibfield  {author} {\bibinfo {author} {\bibfnamefont {T.}~\bibnamefont {Kavitha}}, \bibinfo {author} {\bibfnamefont {K.}~\bibnamefont {Mehlhorn}}, \bibinfo {author} {\bibfnamefont {D.}~\bibnamefont {Michail}},\ and\ \bibinfo {author} {\bibfnamefont {K.~E.}\ \bibnamefont {Paluch}},\ }\bibfield  {title} {\bibinfo {title} {An algorithm for minimum cycle basis of graphs},\ }\href@noop {} {\bibfield  {journal} {\bibinfo  {journal} {Algorithmica}\ }\textbf {\bibinfo {volume} {52}},\ \bibinfo {pages} {333} (\bibinfo {year} {2008})}\BibitemShut {NoStop}%
\bibitem [{\citenamefont {Mahmoud}\ and\ \citenamefont {Ali}(2026{\natexlab{a}})}]{mahmoud2026hqecc_threshold}%
  \BibitemOpen
  \bibfield  {author} {\bibinfo {author} {\bibfnamefont {A.~A.}\ \bibnamefont {Mahmoud}}\ and\ \bibinfo {author} {\bibfnamefont {K.~M.}\ \bibnamefont {Ali}},\ }\href {https://doi.org/10.5281/zenodo.18784824} {\bibinfo {title} {{HQECC-Threshold Simulation Code}}},\ \bibinfo {howpublished} {Zenodo} (\bibinfo {year} {2026}{\natexlab{a}})\BibitemShut {NoStop}%
\bibitem [{\citenamefont {Mahmoud}\ and\ \citenamefont {Ali}(2026{\natexlab{b}})}]{mahmoud_hqecc_github}%
  \BibitemOpen
  \bibfield  {author} {\bibinfo {author} {\bibfnamefont {A.~A.}\ \bibnamefont {Mahmoud}}\ and\ \bibinfo {author} {\bibfnamefont {K.~M.}\ \bibnamefont {Ali}},\ }\href@noop {} {\bibinfo {title} {{HQECC-Threshold GitHub Repository}}} (\bibinfo {year} {2026}{\natexlab{b}}),\ \bibinfo {note} {\url{https://github.com/AhmeedAdelMahmoud/HQECC-Threshold}. Accessed: 2026-03-15}\BibitemShut {NoStop}%
\bibitem [{\citenamefont {Delfosse}\ and\ \citenamefont {Nickerson}(2021)}]{delfosse2021almost}%
  \BibitemOpen
  \bibfield  {author} {\bibinfo {author} {\bibfnamefont {N.}~\bibnamefont {Delfosse}}\ and\ \bibinfo {author} {\bibfnamefont {N.~H.}\ \bibnamefont {Nickerson}},\ }\bibfield  {title} {\bibinfo {title} {Almost-linear time decoding algorithm for topological codes},\ }\href@noop {} {\bibfield  {journal} {\bibinfo  {journal} {Quantum}\ }\textbf {\bibinfo {volume} {5}},\ \bibinfo {pages} {595} (\bibinfo {year} {2021})}\BibitemShut {NoStop}%
\bibitem [{\citenamefont {Huang}\ \emph {et~al.}(2020)\citenamefont {Huang}, \citenamefont {Newman},\ and\ \citenamefont {Brown}}]{huang2020fault}%
  \BibitemOpen
  \bibfield  {author} {\bibinfo {author} {\bibfnamefont {S.}~\bibnamefont {Huang}}, \bibinfo {author} {\bibfnamefont {M.}~\bibnamefont {Newman}},\ and\ \bibinfo {author} {\bibfnamefont {K.~R.}\ \bibnamefont {Brown}},\ }\bibfield  {title} {\bibinfo {title} {Fault-tolerant weighted union-find decoding on the toric code},\ }\href@noop {} {\bibfield  {journal} {\bibinfo  {journal} {Physical Review A}\ }\textbf {\bibinfo {volume} {102}},\ \bibinfo {pages} {012419} (\bibinfo {year} {2020})}\BibitemShut {NoStop}%
\bibitem [{\citenamefont {Old}\ and\ \citenamefont {Rispler}(2023)}]{old2023generalized}%
  \BibitemOpen
  \bibfield  {author} {\bibinfo {author} {\bibfnamefont {J.}~\bibnamefont {Old}}\ and\ \bibinfo {author} {\bibfnamefont {M.}~\bibnamefont {Rispler}},\ }\bibfield  {title} {\bibinfo {title} {Generalized belief propagation algorithms for decoding of surface codes},\ }\href@noop {} {\bibfield  {journal} {\bibinfo  {journal} {Quantum}\ }\textbf {\bibinfo {volume} {7}},\ \bibinfo {pages} {1037} (\bibinfo {year} {2023})}\BibitemShut {NoStop}%
\bibitem [{\citenamefont {Delfosse}(2013)}]{delfosse2013tradeoffs}%
  \BibitemOpen
  \bibfield  {author} {\bibinfo {author} {\bibfnamefont {N.}~\bibnamefont {Delfosse}},\ }\bibfield  {title} {\bibinfo {title} {Tradeoffs for reliable quantum information storage in surface codes and color codes},\ }in\ \href@noop {} {\emph {\bibinfo {booktitle} {2013 IEEE International Symposium on Information Theory}}}\ (\bibinfo {organization} {IEEE},\ \bibinfo {year} {2013})\ pp.\ \bibinfo {pages} {917--921}\BibitemShut {NoStop}%
\bibitem [{\citenamefont {Davydova}\ \emph {et~al.}(2023)\citenamefont {Davydova}, \citenamefont {Tantivasadakarn},\ and\ \citenamefont {Balasubramanian}}]{davydova2023floquet}%
  \BibitemOpen
  \bibfield  {author} {\bibinfo {author} {\bibfnamefont {M.}~\bibnamefont {Davydova}}, \bibinfo {author} {\bibfnamefont {N.}~\bibnamefont {Tantivasadakarn}},\ and\ \bibinfo {author} {\bibfnamefont {S.}~\bibnamefont {Balasubramanian}},\ }\bibfield  {title} {\bibinfo {title} {Floquet codes without parent subsystem codes},\ }\href@noop {} {\bibfield  {journal} {\bibinfo  {journal} {PRX Quantum}\ }\textbf {\bibinfo {volume} {4}},\ \bibinfo {pages} {020341} (\bibinfo {year} {2023})}\BibitemShut {NoStop}%
\bibitem [{\citenamefont {Kesselring}\ \emph {et~al.}(2024)\citenamefont {Kesselring}, \citenamefont {Magdalena de~la Fuente}, \citenamefont {Thomsen}, \citenamefont {Eisert}, \citenamefont {Bartlett},\ and\ \citenamefont {Brown}}]{kesselring2024anyon}%
  \BibitemOpen
  \bibfield  {author} {\bibinfo {author} {\bibfnamefont {M.~S.}\ \bibnamefont {Kesselring}}, \bibinfo {author} {\bibfnamefont {J.~C.}\ \bibnamefont {Magdalena de~la Fuente}}, \bibinfo {author} {\bibfnamefont {F.}~\bibnamefont {Thomsen}}, \bibinfo {author} {\bibfnamefont {J.}~\bibnamefont {Eisert}}, \bibinfo {author} {\bibfnamefont {S.~D.}\ \bibnamefont {Bartlett}},\ and\ \bibinfo {author} {\bibfnamefont {B.~J.}\ \bibnamefont {Brown}},\ }\bibfield  {title} {\bibinfo {title} {Anyon condensation and the color code},\ }\href@noop {} {\bibfield  {journal} {\bibinfo  {journal} {PRX Quantum}\ }\textbf {\bibinfo {volume} {5}},\ \bibinfo {pages} {010342} (\bibinfo {year} {2024})}\BibitemShut {NoStop}%
\bibitem [{\citenamefont {Soares~Jr.}\ and\ \citenamefont {da~Silva}(2018)}]{soares2018hyperbolic}%
  \BibitemOpen
  \bibfield  {author} {\bibinfo {author} {\bibfnamefont {W.~S.}\ \bibnamefont {Soares~Jr.}}\ and\ \bibinfo {author} {\bibfnamefont {E.~B.}\ \bibnamefont {da~Silva}},\ }\bibfield  {title} {\bibinfo {title} {Hyperbolic quantum color codes},\ }\href@noop {} {\bibfield  {journal} {\bibinfo  {journal} {Quantum Information and Computation}\ }\textbf {\bibinfo {volume} {18}},\ \bibinfo {pages} {306} (\bibinfo {year} {2018})}\BibitemShut {NoStop}%
\bibitem [{\citenamefont {Johnson}(1975)}]{johnson1975finding}%
  \BibitemOpen
  \bibfield  {author} {\bibinfo {author} {\bibfnamefont {D.~B.}\ \bibnamefont {Johnson}},\ }\bibfield  {title} {\bibinfo {title} {Finding all the elementary circuits of a directed graph},\ }\href@noop {} {\bibfield  {journal} {\bibinfo  {journal} {SIAM Journal on Computing}\ }\textbf {\bibinfo {volume} {4}},\ \bibinfo {pages} {77} (\bibinfo {year} {1975})}\BibitemShut {NoStop}%
\bibitem [{\citenamefont {Birmel{\'e}}\ \emph {et~al.}(2013)\citenamefont {Birmel{\'e}}, \citenamefont {Ferreira}, \citenamefont {Grossi}, \citenamefont {Marino}, \citenamefont {Pisanti}, \citenamefont {Rizzi},\ and\ \citenamefont {Sacomoto}}]{birmele2013optimal}%
  \BibitemOpen
  \bibfield  {author} {\bibinfo {author} {\bibfnamefont {E.}~\bibnamefont {Birmel{\'e}}}, \bibinfo {author} {\bibfnamefont {R.}~\bibnamefont {Ferreira}}, \bibinfo {author} {\bibfnamefont {R.}~\bibnamefont {Grossi}}, \bibinfo {author} {\bibfnamefont {A.}~\bibnamefont {Marino}}, \bibinfo {author} {\bibfnamefont {N.}~\bibnamefont {Pisanti}}, \bibinfo {author} {\bibfnamefont {R.}~\bibnamefont {Rizzi}},\ and\ \bibinfo {author} {\bibfnamefont {G.}~\bibnamefont {Sacomoto}},\ }\bibfield  {title} {\bibinfo {title} {Optimal listing of cycles and st-paths in undirected graphs},\ }in\ \href@noop {} {\emph {\bibinfo {booktitle} {Proceedings of the twenty-fourth annual ACM-SIAM symposium on Discrete algorithms}}}\ (\bibinfo {organization} {SIAM},\ \bibinfo {year} {2013})\ pp.\ \bibinfo {pages} {1884--1896}\BibitemShut {NoStop}%
\bibitem [{\citenamefont {Gupta}\ and\ \citenamefont {Suzumura}(2021)}]{gupta2021finding}%
  \BibitemOpen
  \bibfield  {author} {\bibinfo {author} {\bibfnamefont {A.}~\bibnamefont {Gupta}}\ and\ \bibinfo {author} {\bibfnamefont {T.}~\bibnamefont {Suzumura}},\ }\bibfield  {title} {\bibinfo {title} {Finding all bounded-length simple cycles in a directed graph},\ }\href@noop {} {\bibfield  {journal} {\bibinfo  {journal} {arXiv preprint arXiv:2105.10094}\ } (\bibinfo {year} {2021})}\BibitemShut {NoStop}%
\end{thebibliography}%

\end{document}